\documentclass[amsmath,amssymb]{revtex4}
\usepackage{graphicx}% Include figure files
\usepackage{amscd,amsmath,amsthm,amsfonts,amssymb}
\usepackage[colorlinks,linkcolor=red,anchorcolor=blue,citecolor=green]{hyperref}
\usepackage{ytableau}
\usepackage{extarrows}
\newtheorem{lem}{Lemma}
\newtheorem{thm}{Theorem}
\makeatletter \@addtoreset{equation}{section} \makeatother

\def\[{\begin{equation}}
\def\]{\end{equation}}

\begin{document}
\title{Rogue-wave and lump patterns associated with the third Painlev\'{e} equation}
\author{%%%% Author details
Bo Yang$^{1}$, Jianke Yang$^{2}$}
%%%%%%%%% Insert author address here
\address{$^{1}$ School of Mathematics and Statistics, Ningbo University, Ningbo 315211, China\\
$^{2}$ Department of Mathematics and Statistics, University of Vermont, Burlington, VT 05401, U.S.A}
\begin{abstract}
We report rogue-wave and lump patterns associated with Umemura polynomials, which arise in rational solutions of the third Painlev\'{e} equation. We first show that in many integrable equations such as the nonlinear Schr\"odinger equation and the Boussinesq equation, when internal parameters of their rogue wave solutions are large and of certain form, then their rogue patterns in the spatial-temporal plane can be asymptotically predicted by root distributions of Umemura polynomials (or equivalently, pole distributions of rational solutions to the third Painlev\'{e} equation). Specifically, every simple root of the Umemura polynomial would induce a fundamental rogue wave whose spatial-temporal location is linearly related to that simple root, while a multiple root of the Umemura polynomial would induce a non-fundamental rogue wave in the $O(1)$ neighborhood of the spatial-temporal origin. Next, we show that in a certain class of higher-order lump solutions of the Kadomtsev-Petviashvili-I (KPI) equation, when their internal parameters are large and of certain form, then their lump patterns at $O(1)$ time can also be predicted asymptotically by root distributions of Umemura polynomials, where simple and multiple roots of the polynomial would give rise to fundamental and non-fundamental lumps in the spatial plane, respectively. These results reveal the importance of the third Painlev\'{e} equation in studies of nonlinear wave patterns. We also report a new transformation which turns bilinear rogue-wave solutions of the nonlinear Schr\"odinger equation to higher-order lump solutions of the KPI equation.
\end{abstract} 
\maketitle

\section{Introduction}
Pattern formation of nonlinear waves is a fascinating and important field of applied mathematics. These patterns represent interesting and distinctive behaviors of nonlinear wave solutions that can often be related to actual physical phenomena, and theoretical analysis of these patterns can lead to deep and universal mathematical objects.

In recent years, many nonlinear wave patterns were linked to pole distributions of solutions to certain Painlev\'{e} equations. Painlev\'{e} equations are certain nonlinear second-order ordinary differential equations in the complex plane possessing the classical Painlev\'{e} property, namely, that their solutions have no movable critical singularities (such as movable branch points or logarithmic singularities), but which are not generally solvable in terms of elementary functions. There are six types of Painlev\'{e} equations, $\mbox{P}_{\mbox{\scriptsize I}}$, $\mbox{P}_{\mbox{\scriptsize II}}$, $\dots$, and $\mbox{P}_{\mbox{\scriptsize VI}}$, which were classified by Painlev\'{e} and Gambier \cite{Painleve1900, Painleve1902, Gambier1910}. These six Painlev\'{e} equations furthermore satisfy the stronger property that all movable singularities of their solutions are poles. An introduction to these Painlev\'{e} equations can be found in \cite{ClarksonNIST}.

So far, nonlinear wave patterns have been linked to pole distributions of solutions to $\mbox{P}_{\mbox{\scriptsize I}}$, $\mbox{P}_{\mbox{\scriptsize II}}$ and $\mbox{P}_{\mbox{\scriptsize IV}}$ (to the authors' best knowledge). Specifically, in the semiclassical nonlinear Schr\"odinger (NLS) equation after wave breaking, a cascade of Peregrine rogue waves appear, and their locations are determined by poles of the tritronqu\'ee solution to $\mbox{P}_{\mbox{\scriptsize I}}$ \cite{Tovbis}. In the semiclassical sine-Gordon equation with initial conditions near the separatrix of a simple pendulum, superluminal (infinite velocity) kinks that appear in the solution were linked to real-axis poles of rational solutions to $\mbox{P}_{\mbox{\scriptsize II}}$ \cite{Miller_sine}. In many integrable equations such as the NLS equation, the Boussinesq equation, the Manakov system and the three-wave interaction system, there exist rogue wave solutions that ``come from nowhere and leave with no trace", i.e., they arise from a uniform background, reach much higher amplitude in local space, and then retreat to the same background again. When one of their internal parameters is large, the rogue field would exhibit geometric patterns formed by fundamental rogue waves, and spatial-temporal locations of these fundamental rogue waves have been linked to pole distributions of rational solutions to $\mbox{P}_{\mbox{\scriptsize II}}$ and $\mbox{P}_{\mbox{\scriptsize IV}}$, or equivalently, to root distributions of Yablonskii-Vorob'ev polynomials and Okamoto polynomials \cite{YangNLS2021, YangOkamoto2023, Yangbook2024}. In the Sasa-Satsuma equation, there exist a class of partial-rogue waves that ``come from nowhere but leave with a trace''. Analytical prediction of these partial-rogue waves has been linked to pole distributions of rational solutions to $\mbox{P}_{\mbox{\scriptsize IV}}$, or equivalently, to root distributions of generalized Okamoto polynomials \cite{YangPartialRogue}. In the Kadomtsev-Petviashvili-I (KPI) equation, there exist higher-order lump solutions that describe anomalous scatterings of fundamental lumps with the same asymptotic velocities. If internal parameters in those higher-order lumps meet certain conditions, then at large time, they would exhibit geometric patterns comprising fundamental lumps, whose spatial locations have been linked to pole distributions of rational solutions to $\mbox{P}_{\mbox{\scriptsize II}}$ and $\mbox{P}_{\mbox{\scriptsize IV}}$,  or equivalently, to root distributions of Yablonskii-Vorob'ev polynomials, generalized Hermite polynomials and generalized Okamoto polynomials (the latter two are special classes of Wronskian-Hermite polynomials) \cite{YangYangKPI, Yangbook2024,Clarkson2003PIV}.

In this paper, we study rogue-wave and lump patterns that are linked to pole distributions of rational solutions to $\mbox{P}_{\mbox{\scriptsize III}}$, or equivalently, to root distributions of Umemura polynomials since those rational solutions of $\mbox{P}_{\mbox{\scriptsize III}}$ are expressed through Umemura polynomials. On rogue patterns, we show that in many integrable equations such as the NLS equation and the Boussinesq equation, when internal parameters of their rogue wave solutions are large and of certain form, then their rogue patterns in the spatial-temporal plane can be asymptotically predicted by root structures of Umemura polynomials, where simple and multiple roots of these polynomials would induce fundamental and non-fundamental rogue wave components in the wavefield, respectively. On lump patterns, we show that in a certain class of higher-order lump solutions of KPI, when their internal parameters are large and of certain form, then their wavefields at $O(1)$ time can also be predicted asymptotically by root structures of Umemura polynomials. We compare these asymptotic predictions to true solutions and demonstrate good agreement between them. These results reveal the importance of $\mbox{P}_{\mbox{\scriptsize III}}$ in studies of rogue wave patterns and lump patterns. Along the way, we also report a new transformation which converts bilinear rogue-wave solutions of the NLS equation to higher-order lump solutions of the KPI equation.

This paper is structured as follows. In Sec.~\ref{secPIII}, we introduce the third Painlev\'{e} equation $\mbox{P}_{\mbox{\scriptsize III}}$ and its associated Umemura polynomials. We also review properties of these polynomials and their root structures. In Sec.~\ref{sec:NLS}, we study rogue wave patterns in the NLS equation and the Boussinesq equation in certain parameter regimes, and show that those patterns can be analytically predicted by root distributions of Umemura polynomials. We also compare these predictions with true rogue patterns. In Sec.~\ref{sec:KP}, we study higher-order lumps of the KPI equation in certain parameter regimes, and show that their lump patterns can be analytically predicted by root distributions of Umemura polynomials as well. A transformation that converts NLS rogue wave solutions to KPI higher-order lump solutions is also presented here. Sec.~\ref{sec:conclusion} summarizes the paper with some additional discussions.

\section{Umemura polynomials associated with the third Painlev\'{e} equation}\label{secPIII}
In this section, we introduce Umemura polynomials associated with the third Painlev\'{e} equation and review their properties.

\subsection{Umemura polynomials in rational solutions of the third Painlev\'{e} equation}

The third Painlev\'{e} equation ($\mbox{P}_{\mbox{\scriptsize III}}$) is
\[\label{PainleveIII}
w''=\frac{(w')^2}{w} -\frac{w'}{z}+\frac{\alpha w^2 + \beta}{z} + \gamma w^3+\frac{\delta}{w},
\]
where the prime denotes differentiation with respect to $z$, and $\alpha,\ \beta,\ \gamma$, $\delta$ are arbitrary complex constants. In the generic case where $\gamma \delta \neq 0$, by rescaling $w$ and $z$, one can set $\gamma=1$ and $\  \delta=-1$ without loss of generality. Then, Eq.~(\ref{PainleveIII}) becomes
\[\label{PainleveIII1}
w''=\frac{(w')^2}{w} -\frac{w'}{z}+\frac{\alpha w^2 + \beta}{z} +  w^3 -\frac{1}{w}.
\]
It has been shown by Gromak \cite{Gromak1984} that this $\mbox{P}_{\mbox{\scriptsize III}}$ equation admits rational solutions if and only if $\alpha+ \varepsilon \beta= 4n$, where $\varepsilon=\pm 1$, $n \in \mathbb{Z}$, and $\mathbb{Z}$ denotes the set of all integers. It has also been shown by Murata \cite{Murata1995} that when this $\mbox{P}_{\mbox{\scriptsize III}}$ admits rational solutions, the number of rational solutions for each $(\alpha, \beta)$ is two or four, and four solutions arise if and only if there exist two integers $m$ and $n$ such that
$\alpha+ \beta= 4n$ and $\alpha- \beta= 4m$.

Rational solutions for the case of $\alpha+\beta=4n$ were derived by Umemura \cite{UmemuraAMS2001} and expressed through special polynomials in $1/z$. Clarkson \cite{Clarkson2003PIII} showed that these rational solutions can also be expressed through polynomials in $z$. Specifically, he introduced Umemura polynomials $U_n(z; \mu)$ through the recurrence relation
\[\label{polySn}
U_{n+1} U_{n-1}=-z\left[U_n U_n'' -  \left(U_n'\right)^2  \right]-U_n U_n'+(z+\mu)U_n^2,
\]
with $U_{-1}(z;\mu)=U_{0}(z;\mu)=1$ and $\mu$ a free complex parameter. The polynomiality of $U_n(z; \mu)$ in $z$ is far from obvious from this recurrence relation but has been shown to be true \cite{Kajiwara1999,Clarkson2023PIII}. In terms of these Umemura polynomials, the function
\[ \label{wn}
w_n\equiv w(z; \alpha_n, \beta_n)=1+\frac{d}{dz} \ln \left\{\frac{U_{n-1}(z;\mu-1)}{U_n(z;\mu)}\right\} \equiv \frac{U_n(z;\mu-1)U_{n-1}(z;\mu)}{U_n(z;\mu) U_{n-1}(z;\mu-1)}
\]
will satisfy the $\mbox{P}_{\mbox{\scriptsize III}}$ equation (\ref{PainleveIII1}) with $\alpha_n=2n+2\mu-1$, $\beta_n=2n-2\mu+1$, and $n \in \mathbb{Z}$ (see \cite{Clarkson2003PIII}). The other rational solution of $\mbox{P}_{\mbox{\scriptsize III}}$ for $\alpha_n=2n+2\mu-1$ and $\beta_n=2n-2\mu+1$ can be obtained from the above $w_n$ solution through the B\"acklund transformation that if $w(z; \alpha, \beta)$ is a solution of the $\mbox{P}_{\mbox{\scriptsize III}}$ equation (\ref{PainleveIII1}), then so is $-1/w(z; \beta, \alpha)$ \cite{Clarkson2003PIII,Kajiwara2003}.
Rational solutions of the $\mbox{P}_{\mbox{\scriptsize III}}$ equation (\ref{PainleveIII1}) for the case of $\alpha-\beta=4n$ can be obtained from those of $\alpha+\beta=4n$ through another B\"acklund transformation that if $w(z; \alpha, \beta)$ is a solution of the $\mbox{P}_{\mbox{\scriptsize III}}$ equation (\ref{PainleveIII1}), then so is ${\rm i}w({\rm i}z; \alpha, -\beta)$ \cite{Clarkson2003PIII,Kajiwara2003}.

\subsection{Determinant form of Umemura polynomials}
For our studies of rogue-wave and lump patterns, we need to put these Umemura polynomials in a certain determinant form. One determinant expression has been obtained by Kajiwara and Masuda~\cite{Kajiwara1999} (see also \cite{Clarkson2003PIII}). Let $\hat{p}_{k}(z;\mu)$ be the Schur polynomial defined by the generating function
\begin{equation}\label{polypkzmu}
\sum_{k=0}^{\infty} \hat{p}_{k}(z;\mu) \epsilon^k =(1+ \epsilon)^\mu \exp\left( z \epsilon \right),
\end{equation}
with $\hat{p}_k(z;\mu)\equiv 0$ for $k<0$.  Then,  the Umemura polynomial $U_{n}(z;\mu)$ for a positive integer $n$ is given by the $n \times n$ determinant \cite{Kajiwara1999,Clarkson2003PIII}
\[\label{DetUmemura0}
U_n(z;\mu)=c_n \left| \begin{array}{cccc}
         \hat{p}_1(z;\mu) & \hat{p}_0(z;\mu) & \cdots &  \hat{p}_{2-n}(z;\mu) \\
         \hat{p}_3(z;\mu) & \hat{p}_2(z;\mu) & \cdots & \hat{p}_{4-n}(z;\mu) \\
        \vdots& \vdots & \ddots & \vdots \\
         \hat{p}_{2n-1}(z;\mu) & \hat{p}_{2n-2}(z;\mu) & \cdots & \hat{p}_n(z;\mu)
       \end{array} \right|,
\]
where $c_n=\prod_{j=1}^n (2j+1)^{n-j}$.

However, to link rogue and lump patterns to Umemura polynomials, we need to modify the above Schur polynomial $\hat{p}_k(z; \mu)$. The reason for this modification will be seen in later sections. To modify, we first rewrite Eq.~(\ref{polypkzmu}) as
\begin{equation}\label{polypkzmu1}
\sum_{k=0}^{\infty} \hat{p}_{k}(z;\mu) \epsilon^k = \exp\left[ z \epsilon + \mu\log(1+\epsilon) \right] = \exp\left( z \epsilon + \mu \sum_{j=1}^{\infty} \frac{(-1)^{j-1}}{j} \epsilon^{j} \right).
\end{equation}
From this equation, we can see that the Umemura polynomial $U_{n}(z;\mu)$ in Eq.~(\ref{DetUmemura0}) can be defined alternatively as
\[\label{DetUmemura}
U_n(z;\mu)=c_n \left| \begin{array}{cccc}
         p_1(z;\mu) & p_0(z;\mu) & \cdots &  p_{2-n}(z;\mu) \\
         p_3(z;\mu) & p_2(z;\mu) & \cdots & p_{4-n}(z;\mu) \\
        \vdots& \vdots & \ddots & \vdots \\
         p_{2n-1}(z;\mu) & p_{2n-2}(z;\mu) & \cdots & p_n(z;\mu)
       \end{array} \right|,
\]
where the new Schur polynomial $p_k(z; \mu)$ is given by the following generating function
\begin{equation}\label{polypkzmu2}
\sum_{k=0}^{\infty} p_{k}(z;\mu) \epsilon^k = \exp\left( z \epsilon + \mu \sum_{j=1}^{\infty} \frac{1}{2j-1} \epsilon^{2j-1} \right)= \exp\left( (z+\mu) \epsilon +  \sum_{j=1}^{\infty} \frac{\mu}{2j+1} \epsilon^{2j+1} \right),
\end{equation}
in which all the even-power terms of $\epsilon$ in the right-side exponent of Eq.~(\ref{polypkzmu1}) have been dropped. The reason is due to the structure of the determinant in Eq.~(\ref{DetUmemura0}) with index jumps of two. For such a determinant, it is easy to show that the partial derivative of $U_n(z;\mu)$ with respect to the even-power coefficient of $\epsilon$ on the right side of Eq.~(\ref{polypkzmu1}) is zero. Thus, those even-power terms of $\epsilon$ do not affect $U_n(z;\mu)$ and can be dropped, leading to the equivalent generating function (\ref{polypkzmu2}). In the rest of this article, we will use the modified determinant form (\ref{DetUmemura}) for the Umemura polynomial $U_{n}(z;\mu)$. This determinant is a Wronskian since we can see from Eq.~(\ref{polypkzmu2}) that $p'_{j}(z; \mu)=p_{j-1}(z; \mu)$, where the prime denotes differentiation with respect to $z$.

\subsection{Properties of Umemura polynomials and their roots}
The $U_{n}(z;\mu)$ polynomial is monic with degree $n(n+1)/2$, which can be seen from its determinant expression (\ref{DetUmemura}). It also admits the following symmetry property~\cite{Clarkson2003PIII}:
\[\label{SymmetrySn}
U_n(-z; -\mu)=(-1)^{n(n+1)/2}U_n(z; \mu),
\]
which can be proved using the recurrence relation (\ref{polySn}). Here, the missing factor of $(-1)^{n(n+1)/2}$ in \cite{Clarkson2003PIII} has been added.

The first few Umemura polynomials are
\begin{eqnarray*}
&& U_1(z;\mu)=z+\mu, \\
&& U_2(z;\mu)=(z+\mu)^3-\mu, \\
&& U_3(z;\mu)=(z+\mu)^6 - 5\mu(z+\mu)^3 + 9\mu(z+\mu)-5\mu ^2,  \\
&& U_4(z;\mu)=(z+\mu)^{10}-15\mu(z+\mu)^7+63 \mu(z+\mu)^5-225 \mu (z+\mu)^3+315 \mu ^2(z+\mu)^2-175 \mu^3(z+\mu)+36 \mu ^2,  \\
&& U_5(z;\mu)=(z+\mu )^{15}-35 \mu  (z+\mu )^{12}+252 \mu  (z+\mu )^{10}+175 \mu ^2 (z+\mu )^9-2025 \mu(z+\mu)^8 + 945\mu^2(z+\mu)^7\\
&&\hspace{1.5cm}   -1225 \mu  \left(\mu ^2-9\right) (z+\mu )^6-26082 \mu ^2 (z+\mu )^5+33075 \mu ^3 (z+\mu )^4-350 \mu ^2 \left(36+35\mu^2\right) (z+\mu )^3\\
&&\hspace{1.5cm}   + 11340 \mu ^3 (z+\mu )^2+225 \mu ^2 \left(36-49 \mu ^2\right) (z+\mu )-5796 \mu ^3+6125 \mu ^5.
\end{eqnarray*}

Root structures of Umemura polynomials are important to us, since we will link them to rogue-wave and lump patterns in the later text. Root properties of these polynomials have been studied extensively in the literature
\cite{Clarkson2003PIII,Amdeberhan2006,Clarkson2023PIII}. First, Clarkson \cite{Clarkson2003PIII} and Amdeberhan \cite{Amdeberhan2006} reported a simple formula for the discriminant of $U_{n}(z;\mu)$, from which one can see that this polynomial for any positive integer $n$ admits multiple roots if and only if $\mu=0, \pm 1, \dots, \pm (n-2)$. Amdeberhan \cite{Amdeberhan2006} and Clarkson, Law \& Lin \cite{Clarkson2023PIII} also reported a simple formula for $U_{n}(0;\mu)$, from which one can see zero is a root of $U_{n}(z;\mu)$ when $\mu=0, \pm 1, \dots, \pm (n-1)$. In addition, Clarkson, Law \& Lin \cite{Clarkson2023PIII} proved that all nonzero roots of $U_{n}(z;\mu)$ are simple, and the multiplicity of the zero root is $(n-|\mu|)(n-|\mu|+1)/2$ when $\mu=0, \pm 1, \dots, \pm (n-1)$. Putting all these results together, we have the following lemma on the root properties of Umemura polynomials.

\begin{lem} (\cite{Clarkson2003PIII,Amdeberhan2006,Clarkson2023PIII})   \label{Lemma1}
For the Umemura polynomial $U_{N}(z;\mu)$ with $N$ a positive integer and $\mu$ a complex parameter,
\begin{enumerate}
\item if $\mu$ is equal to one of $0, \pm 1, \dots, \pm (N-1)$, then $U_{N}(z;\mu)$ has a zero root of multiplicity $N_0(N_0+1)/2$, where $N_0=N-|\mu|$. The nonzero roots are all simple, and their number is $ N_p=N(N+1)/2-N_0(N_0+1)/2$;
\item if $\mu\ne 0, \pm 1, \dots, \pm (N-1)$, then $U_{N}(z;\mu)$ has $N(N+1)/2$ nonzero simple roots only.
\end{enumerate}
\end{lem}

This lemma gives very clear results regarding multiplicities of the zero and nonzero roots in Umemura polynomials.  These clear root results will lead to clear rogue-wave and lump pattern predictions associated with Umemura polynomials that we will derive later in this paper. For other special polynomials, such clear root-multiplicity results are not always known. For instance, for Yablonskii-Vorob'ev hierarchy polynomials, Okamoto hierarchy polynomials and Wronskian-Hermite polynomials, the simplicity of their nonzero roots is still a conjecture \cite{Clarkson2003-II,YangOkamoto2023,Felder2012}. Such uncertainty in root results adversely affected the clarity of rogue-wave and lump pattern predictions associated with those polynomials in \cite{YangNLS2021,YangOkamoto2023,YangYangKPI}.

In Fig.~\ref{fig1}, we show six root structures of the Umemura polynomial $U_{N}(z;\mu)$ for $N=5$ at various $\mu$ values. The upper row are root structures at $\mu=1, 2$ and 3. At these integer $\mu$ values, the roots contain $|\mu|$ concentric arcs of simple nonzero roots, plus a multiple zero root at the arc center. The lower row are root structures at $1/100, -8\textrm{i}$ and $7/4$. At these noninteger $\mu$ values, all roots are nonzero and simple, and they exhibit various shapes such as two concentric rings with a center (panel (d)), a triangle with a curved-in bottom (panel (e)), and two arcs plus a pentagon with a center (panel (f)). Most of these root structures have been shown in \cite{Clarkson2003PIII}, but the double-ring and curved triangle ones in panels (d, e) seem new. Root structures of $U_{N}(z;\mu)$ at large $N$ have been studied by Bothner, Miller \& Sheng \cite{Bothner2018} and Bothner \& Miller \cite{Bothner2020}.

\begin{figure}[htb]
\begin{center}
\includegraphics[scale=0.60, bb=420 0 325 530]{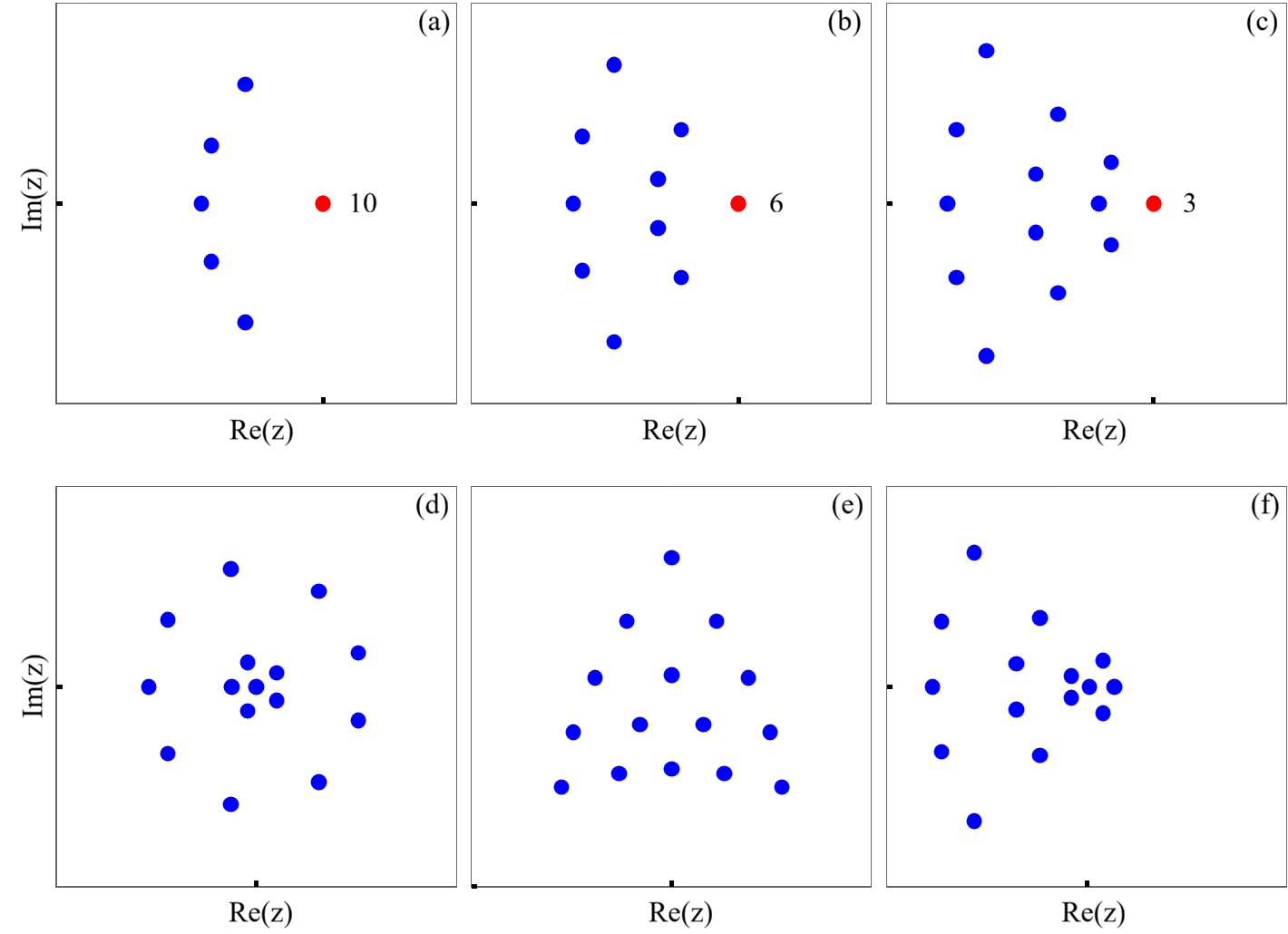}
\caption{Root structures of Umemura polynomials $U_{5}(z;\mu)$ at six $\mu$ values. Upper row: $\mu=1,2,3$ (from left to right). Lower row: $\mu=\frac{1}{100}, -8\textrm{i}, \frac{7}{4}$ (from left to right).  A simple root is marked by a blue dot, and a multiple zero root is marked by a red dot with its multiplicity indicated by a number beside it. In all upper panels, $-8 \leq \Re(z) \leq 4,\  -6 \leq  \Im(z) \leq 6$. In lower panels,
$-3 \leq \Re(z), \Im(z)  \leq 3$ in (d); $-9 \leq \Re(z) \leq 9,\ 0 \leq  \Im(z) \leq 18$ in (e); $-6 \leq \Re(z), \Im(z) \leq6$ in (f). Here, $\Re$ and $\Im$ represent the real and imaginary parts of a complex number, respectively. } \label{fig1}
\end{center}
\end{figure}

\subsection{Relations between Umemura polynomials and Adler-Moser polynomials}
Umemura polynomials $U_{n}(z;\mu)$ are related to certain special Adler-Moser polynomials. Adler-Moser polynomials were proposed by Adler and Moser \cite{Adler_Moser1978}, who expressed rational solutions of the Korteweg-de Vries equation in terms of these polynomials. These polynomials can be expressed through a determinant \cite{Adler_Moser1978,Clarkson2009}
\begin{equation}\label{AdlerMoserdef}
\Theta_{n}(z) = c_{n} \left| \begin{array}{cccc}
         \theta_{1}(z) &\theta_{0}(z) & \cdots &  \theta_{2-n}(z) \\
         \theta_{3}(z) & \theta_{2}(z) & \cdots &  \theta_{4-n}(z) \\
        \vdots& \vdots & \vdots & \vdots \\
         \theta_{2n-1}(z) & \theta_{2n-2}(z) & \cdots &  \theta_{n}(z)
       \end{array} \right|,
\end{equation}
where $\theta_{k}(z)$ are Schur polynomials defined by
\begin{equation}\label{AMthetak}
\sum_{k=0}^{\infty} \theta_k(z) \epsilon^k = \exp\left( z \epsilon + \sum_{j=1}^{\infty} \kappa_j \epsilon^{2j+1} \right),
\end{equation}
$\theta_{k}(z)\equiv 0$ if $k<0$, $c_{n}$ is given below Eq.~(\ref{DetUmemura0}), and $\kappa_j \hspace{0.04cm} (j\ge 1)$ are arbitrary complex constants. Adler-Moser polynomials reduce to the Yablonskii-Vorob'ev polynomial hierarchy when all $\kappa_j$ constants are set to zero except for one of them \cite{YangNLS2021}.

By comparing the determinant expression (\ref{DetUmemura}) of the Umemura polynomial to the above determinant expression (\ref{AdlerMoserdef}) of the Adler-Moser polynomial, we can easily see that the Umemura polynomial $U_{n}(z;\mu)$ is related to the special Adler-Moser polynomial $\Theta_{n}(z; \kappa_1, \kappa_2, \dots, \kappa_{n-1})$ with  $\kappa_j= \mu/(2j+1), 1\le j\le n-1$ as
\[ \label{UmemuraAD}
U_{n}(z;\mu) = \Theta_{n}(z+\mu).
\]
Thus, a root $z_0$ of the Umemura  polynomial corresponds to a root $z_0+\mu$ of this special Adler-Moser polynomial.

Due to this $\mu$-shift of the roots between Adler-Moser and Umemura polynomials, a zero root of the Umemura polynomial becomes a nonzero root of the Adler-Moser polynomial, and vice versa. This is an important distinction in the context of rogue-wave and lump studies. More will be said on this in Sec.~\ref{secDiff} later.

In the next two sections, we will link these Umemura polynomials of the $\mbox{P}_{\mbox{\scriptsize III}}$ equation to certain rogue-wave and lump patterns of integrable systems.

\section{Rogue wave patterns associated with Umemura polynomials in the NLS equation} \label{sec:NLS}

The NLS equation
\[ \label{NLS-2020}
\textrm{i} u_{t} +  \frac{1}{2}u_{xx}+ |u|^2 u=0
\]
arises in numerous physical situations such as water waves and optics~\cite{Benney,Zakharov,Hasegawa}. Its rogue-wave solutions have been derived in many articles before \cite{Peregrine1983,AAS2009,DGKM2010,KAAN2011,GLML2012,OhtaJY2012,YangNLS2021}. Some of those solutions have also been experimentally observed in water tanks \cite{Tank1,Tank2, Tank3, Tank4}, optical fibers \cite{Fiber1,Fiber1b}, plasma \cite{Plasma}, Bose-Einstein condensates \cite{RogueBEC}, acoustics \cite{Acoustics}, and so on.

The simplest expressions of general NLS rogue waves were derived by the bilinear method and presented in Ref.~\cite{YangNLS2021}. To present these bilinear rogue waves, we first introduce elementary Schur polynomials $S_j(\mbox{\boldmath $x$})$, where $\emph{\textbf{x}}=\left( x_{1}, x_{2}, \ldots \right)$. These polynomials are defined by
\begin{equation} \label{Schurdef}
\sum_{j=0}^{\infty}S_j(\mbox{\boldmath $x$}) \epsilon^j
=\exp\left(\sum_{i=1}^{\infty}x_i \epsilon^i\right).
\end{equation}
For example,
\[
S_0(\mbox{\boldmath $x$})=1, \quad S_1(\mbox{\boldmath $x$})=x_1,
\quad S_2(\mbox{\boldmath $x$})=\frac{1}{2}x_1^2+x_2, \quad
S_3(\mbox{\boldmath $x$})=\frac{1}{6}x_1^3+x_1x_2+x_3,
\]
and so on. We also define $S_j(\mbox{\boldmath $x$})\equiv 0$ when $j<0$.

Under these notations, compact expressions of NLS rogue waves of bilinear form are given by the following lemma.
\begin{lem}  \label{Lemma2}
( \cite{YangNLS2021}) General rogue waves of the NLS equation (\ref{NLS-2020}) under boundary conditions of $u(x, t)\to e^{\textrm{i}t}$ as $x, t\to \pm \infty$ are
\[ \label{uNform}
u_N(x,t)=\frac{\sigma_{1}}{\sigma_{0}}e^{\textrm{i}t},
\]
where $N$ is a positive integer representing the order of the rogue wave,
\begin{equation} \label{sigma_n}
\sigma_{n}=
\det_{\begin{subarray}{l}
1\leq i, j \leq N
\end{subarray}}
\left( \phi_{2i-1,2j-1}^{(n)} \right),
\end{equation}
\begin{equation} \label{phiijNLS}
\phi_{i,j}^{(n)}=\sum_{\nu=0}^{\min(i,j)} \frac{1}{4^{\nu}} \hspace{0.06cm} S_{i-\nu}\left[(\textbf{\emph{x}}^{+}(n) +\nu \textbf{\emph{s}}+\textbf{\emph{a}}\right]  \hspace{0.06cm} S_{j-\nu}\left[\textbf{\emph{x}}^{-}(n) + \nu \textbf{\emph{s}}+\textbf{\emph{a}}^*\right],
\end{equation}
vectors $\textbf{\emph{x}}^{\pm}(n)=( x_{1}^{\pm}, 0, x_{3}^{\pm}, 0, \cdots)$ are defined by
\[ \label{defxk}
x_{1}^{\pm}=x \pm \textrm{i} t \pm n,  \quad x_{2j+1}^{\pm}= \frac{x\pm 2^{2j} (\textrm{i} t)}{(2j+1)!},  \quad j\ge 1,
\]
the asterisk * represents complex conjugation, $\textbf{\emph{s}}=(0, s_2, 0, s_4, \cdots)$ are coefficients from the expansion
\[ \label{sexpand}
\sum_{j=1}^{\infty} s_{j}\lambda^{j}=\ln \left[\frac{2}{\lambda}  \tanh \left(\frac{\lambda}{2}\right)\right],
\]
$\textbf{\emph{a}}=(a_1, 0, a_3, 0, \dots, a_{2N-1})$ is the vector of internal parameters, and $a_{1}, a_{3}, a_{5}, \cdots, a_{2N-1}$ are free complex constants.
\end{lem}

When $N=1$ and $a_1=0$, the above solution is $u_1(x,t)=\hat{u}_1(x, t)\hspace{0.04cm} e^{\textrm{i}t}$, where
\begin{equation} \label{Pere}
\hat{u}_1(x, t)=1- \frac{4(1+2\textrm{i}t)}{1+4x^2+4t^2}.
\end{equation}
This is the fundamental rogue wave in the NLS equation that was discovered by Peregrine in \cite{Peregrine1983} and is now called the Peregrine wave in the literature. This wave has a single hump of amplitude 3, flanked by two dips on the two sides of the $x$ direction.

For higher $N$ values, if all internal parameters $(a_1, a_3, \dots, a_{2N-1})$ are zero, then the resulting rogue wave $u_{N}(x, t)$ would have a central peak of amplitude $2N+1$ at the spatial-temporal origin $(x, t)=(0,0)$, flanked by gradually weakening lower peaks and valleys surrounding the central peak \cite{AAS2009, Clarkson2010, He2017, Miller2020, Yangbook2024}. This rogue wave is believed to have the highest possible peak amplitude among rogue waves of the $N$-th order and is thus called a super rogue wave in the literature. When some of the internal parameters $(a_1, a_3, \dots, a_{2N-1})$ are large in magnitude, the solution $u_N(x, t)$ would split into a number of Peregrine and lower-order rogue waves which form interesting patterns \cite{Akhmediev2013,He2013,YangNLS2021,Yang2024NLS,Ling2025,Yang2025NLS}.

Rogue patterns in the NLS equation that are associated with Umemura polynomials of the $\mbox{P}_{\mbox{\scriptsize III}}$ equation would arise when internal parameters $a_1, a_3, \cdots, a_{2N-1}$ in the $u_N(x,t)$ solution (\ref{uNform}) are large and of the following form
\[
a_{2j-1}=\frac{\mu}{2j-1} \hspace{0.04cm} A^{2j-1}, \quad 1\le j\le N,
\]
i.e.,
\[ \label{acond2}
a_1=\mu A, \quad a_3=\frac{\mu}{3}A^3, \quad a_5=\frac{\mu}{5}A^5, \quad  \dots, \quad a_{2N-1}=\frac{\mu}{2N-1}A^{2N-1},
\]
where $\mu$ is a free $O(1)$ complex constant, and $A$ is another free complex constant with large modulus, i.e., $|A|\gg 1$. Rogue patterns for this set of internal parameters can be predicted by root structures of the Umemura polynomial $U_N(z; \mu)$, which we will establish later in this section. Thus, we will call (\ref{acond2}) Umemura-type parameters in this article.  But before we establish this fact, we need to make an important clarification first.

\subsection{Differences in rogue patterns between Umemura-type parameters and Adler-Moser-type parameters} \label{secDiff}
In our recent work \cite{Yang2024NLS,Yang2025NLS}, we determined NLS rogue patterns under parameter conditions
\[ \label{acondAD}
a_1=0, \quad a_{2j+1}=\kappa_j \hspace{0.04cm} A^{2j+1},  \quad 1\le j\le N-1,
\]
where  $A$ is a free large positive constant, and $(\kappa_1, \kappa_2,  \dots, \kappa_{N-1})$ are arbitrary $O(1)$ complex constants not being all zero (here $A$ can also be a complex constant but is not necessary since $\kappa_j$ constants are allowed to be complex). In this parameter regime, we showed that the resulting rogue patterns are determined by roots of Adler-Moser polynomials $\Theta_{N}(z; \kappa_1, \kappa_2,  \dots, \kappa_{N-1})$. Thus, we will call these parameters (\ref{acondAD}) as Adler-Moser-type parameters.

If one chooses $\kappa_j=\mu/(2j+1)$ for $1\le j\le N-1$, they will see that the resulting Adler-Moser-type parameters (\ref{acondAD}) are the same as our current Umemura-type parameters (\ref{acond2}) except for the $a_1$ value. One could argue that the nonzero $a_1$ value in Umemura-type parameters  (\ref{acond2}) can be normalized to zero through a shift of the $(x, t)$ axes in the $u_N(x, t)$ solution (\ref{uNform}). Thus, they could be tempted to conclude that rogue patterns under Umemura-type parameters (\ref{acond2}) would just be special cases of the more general Adler-Moser-type parameters (\ref{acondAD}) that we have thoroughly studied in \cite{Yang2024NLS,Yang2025NLS}. To make this argument even stronger, they could also point to the simple root relation (\ref{UmemuraAD}) between Adler-Moser polynomials and Umemura polynomials, and argue that since Umemura polynomials can be viewed as special Adler-Moser polynomials, rogue patterns associated with Umemura polynomials then should quickly follow from those of Adler-Moser polynomials in \cite{Yang2024NLS,Yang2025NLS}. In particular, the rogue-wave field corresponding to a root $z_0$ of the Umemura polynomial $U_{N}(z;\mu)$ should just be that corresponding to the shifted root $z_0+\mu$ of the special Adler-Moser polynomial mentioned above. In short, one could be tempted to conclude that rogue patterns under the current Umemura-type parameters (\ref{acond2}) are just special cases of those under Adler-Moser-type parameters in \cite{Yang2024NLS,Yang2025NLS} and there is nothing new here.

The above argument on triviality of Umemura-type parameters (\ref{acond2}) is incorrect. We will first use an example to illustrate its erroneousness, followed by a detailed explanation on where that argument went wrong.

In this example, we take $N=5$ and two sets of parameters
\[   \label{para1a}
a_1=\mu A, \quad a_3=\frac{\mu}{3}A^3, \quad a_5=\frac{\mu}{5}A^5, \quad a_7=\frac{\mu}{7}A^7,  \quad a_9=\frac{\mu}{9}A^9,
\]
and
\[ \label{para1b}
a_1=0, \quad a_3=\frac{\mu}{3}A^3, \quad a_5=\frac{\mu}{5}A^5, \quad a_7=\frac{\mu}{7}A^7,  \quad a_9=\frac{\mu}{9}A^9,
\]
with $A=6$ and $\mu=2$. The first set of parameters (\ref{para1a}) are Umemura-type (\ref{acond2}), while the second set Adler-Moser-type (\ref{acondAD}). For these two sets of internal parameters, the NLS rogue waves $|u_{5}(x, t)|$ from Eq.~(\ref{uNform}) are plotted in the left and right panels of Fig.~\ref{figAD}, respectively. It is seen that in both panels, there are nine Peregrine waves located on two concentric arcs in the left side of the wave field. The difference between the two panels is in the center region of the two concentric arcs. In the left panel, that center region hosts a third-order super-rogue wave, while in the right panel, the center region hosts a triangular cluster of six Peregrine waves instead.

\begin{figure}[htb]
\begin{center}
\includegraphics[scale=0.35, bb=350 000 650 550]{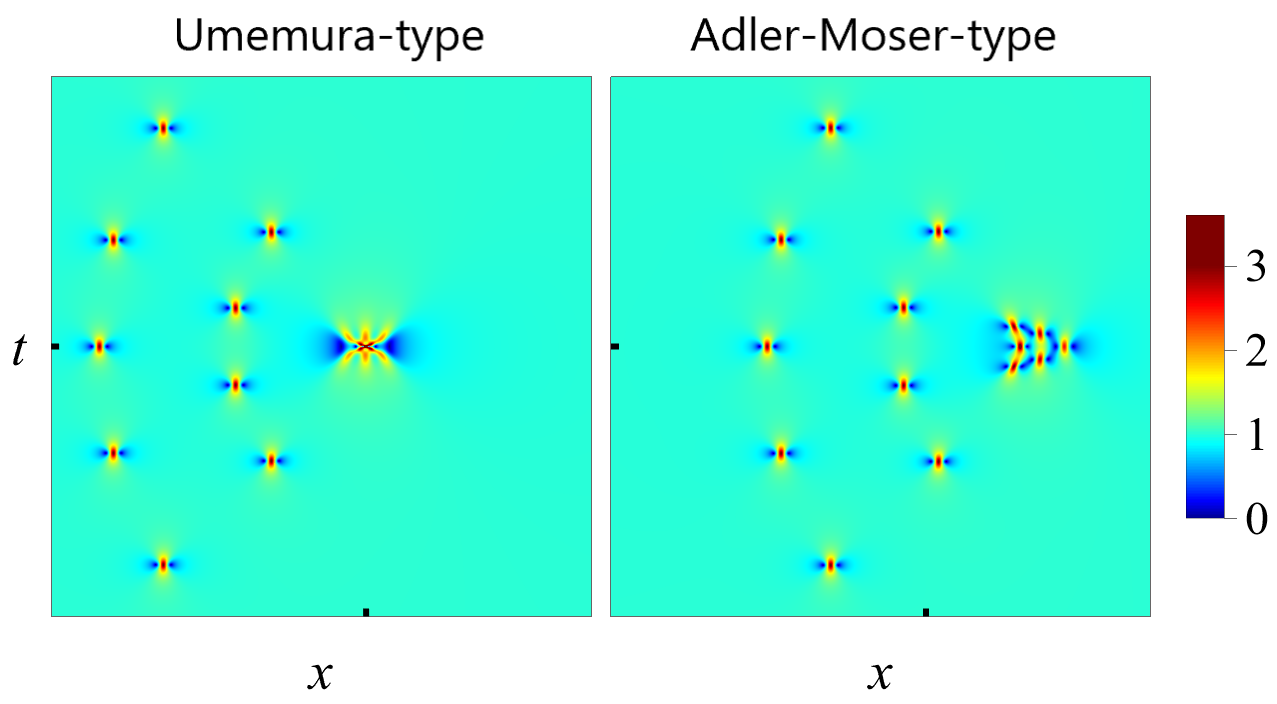}
\caption{NLS rogue waves $|u_{5}(x, t)|$ from Eq.~(\ref{uNform}) for Umemura-type parameters (\ref{para1a}) (left) and Adler-Moser-type parameters (\ref{para1b}) (right) with $A=6$ and $\mu=2$. In both panels, the $(x, t)$ intervals are $-35\leq x \leq 25,\ -30\leq t \leq 30$.  For the left panel,  the color bar does not show the full solution range. } \label{figAD}
\end{center}
\end{figure}

Fig.~\ref{figAD} clearly shows that rogue patterns for Umemura-type parameters (\ref{para1a}) are not the same as those for the corresponding Adler-Moser-type parameters (\ref{para1b}). Thus, the notion of NLS rogue patterns under Umemura-type parameters being simply special cases of rogue patterns of Adler-Moser-type parameters from Refs.~\cite{Yang2024NLS,Yang2025NLS} is incorrect.

Where does the argument for triviality of Umemura-type parameters (\ref{para1a}) go wrong? First, if one argues this triviality on the ground that the nonzero parameter $a_1$ in Eq.~(\ref{para1a}) can be normalized to zero through a shift of the $(x, t)$ axes, the error is that, while this $(x, t)$-shift can indeed make $a_1$ zero, but simultaneously it changes $a_3$, $a_5$, $a_7$ and $a_9$ as well. Specifically, one can see from Eq.~(\ref{defxk}) that this $(x, t)$-shift would change $(a_1, a_3, a_5, a_7, a_9)$ to
\begin{eqnarray*}
&& a_1=0, \quad a_3=\frac{\mu}{3}A^3-\frac{\mu}{3!}\left[\Re(A)+\textrm{i}2^2\Im(A)\right],
\quad a_5=\frac{\mu}{5}A^5-\frac{\mu}{5!}\left[\Re(A)+\textrm{i}2^4\Im(A)\right],      \\
&& a_7=\frac{\mu}{7}A^7-\frac{\mu}{7!}\left[\Re(A)+\textrm{i}2^6\Im(A)\right], \quad
a_9=\frac{\mu}{9}A^9-\frac{\mu}{9!}\left[\Re(A)+\textrm{i}2^8\Im(A)\right],
\end{eqnarray*}
where $\Re$ and $\Im$ represent the real and imaginary parts of a complex number, respectively. These changed parameters are clearly not Adler-Moser type (\ref{acondAD}) or (\ref{para1b}). Thus, one cannot infer rogue patterns under Umemura-type parameters (\ref{para1a}) from those of Adler-Moser parameters (\ref{acondAD}) or (\ref{para1b}) in \cite{Yang2024NLS,Yang2025NLS}. Second, if one argues this triviality on the ground that Umemura polynomials are related to special Adler-Moser polynomials through a simple $z$-shift (\ref{UmemuraAD}), the error is that, we know from Ref.~\cite{Yang2025NLS} that rogue predictions for zero and nonzero multiple roots of the associated special polynomials can be very different. A zero multiple root of the Umemura polynomial under parameters (\ref{para1a}) would correspond to a nonzero multiple root of the corresponding Adler-Moser polynomial. Then, if one uses the rogue field for the nonzero multiple root of the Adler-Moser polynomial to predict the rogue field for the zero multiple root of the Umemura polynomial, it would be wrong, as Fig.~\ref{figAD} clearly shows.

The conclusion from Fig.~\ref{figAD} and the above explanations is that, Umemura-type parameters (\ref{acond2}) is a new parameter regime which features its own rogue patterns that cannot be entirely inferred from rogue patterns of Adler-Moser-type parameters in earlier works \cite{Yang2024NLS,Yang2025NLS}.

We need to point out that, the attempt of using rogue fields of Adler-Moser-type parameters (\ref{acondAD}) to infer rogue fields of Umemura-type parameters (\ref{acond2}) is incorrect but not entirely wrong. As one can see from Fig.~\ref{figAD}, the error of this attempt is only on the portion of the rogue fields that are induced by multiple roots of the underlying Umemura and Adler-Moser polynomials, which lie at the center region of the two arcs in Fig.~\ref{figAD}. The rest of the rogue fields comprising nine lumps (Peregrine waves) on two concentric arcs at the left sides of the panels are induced by simple roots of those polynomials. For those wave fields, one can indeed infer one from another, as can be seen from Fig.~\ref{figAD}.

\subsection{Rogue patterns under Umemura-type parameters} \label{secTheorem1}
Since rogue patterns for Umemura-type parameters (\ref{acond2}) cannot be inferred entirely from those for Adler-Moser-type parameters (\ref{acondAD}) in \cite{Yang2024NLS,Yang2025NLS}, we need to derive rogue patterns for Umemura-type parameters (\ref{acond2}) then. Our analytical predictions for these rogue patterns are summarized in the following theorem.

\begin{thm} \label{Theorem1}
For the NLS rogue wave $u_N(x, t)$ in Eq.~(\ref{uNform}) under Umemura-type parameters (\ref{acond2}), asymptotic predictions for the rogue field at large $|A|$ are as follows.
\begin{enumerate}
\item
If $\mu$ is equal to one of $0, \pm 1, \dots, \pm (N-1)$, then the rogue wave $u_N(x, t)$ asymptotically splits into $N_p$ fundamental (Peregrine) rogue waves located far away from the spatial-temporal origin, plus a $N_0$-th order super rogue wave in the $O(1)$ neighborhood of the origin, where $N_p$ and $N_0$ are as given in Lemma~\ref{Lemma1}. These Peregrine waves are $\hat{u}_1(x-\hat{x}_{0}, t-\hat{t}_{0}) \hspace{0.05cm} e^{\textrm{i}t}$, where $\hat{u}_1(x, t)$ is as given in Eq.~(\ref{Pere}), and their spatial-temporal locations $(\hat{x}_{0}, \hat{t}_{0})$ are given by
\begin{eqnarray}
&&\hat{x}_{0}+\textrm{i}\hspace{0.05cm}\hat{t}_{0}=z_{0} A, \label{x0t0}
\end{eqnarray}
with $z_{0}$ being each of the $N_p$ simple nonzero roots of the Umemura polynomial $U_{N}(z;\mu)$. The  $N_0$-th order super rogue wave $u_{N_0}(x, t)$ is given by Eq.~(\ref{uNform}) with all its internal parameters $(\hat{a}_1, \hat{a}_3, \cdots, \hat{a}_{2N_0-1})$ being zero. The errors of these Peregrine and $N_0$-th order super rogue wave approximations are $O(|A|^{-1})$. If $(x,t)$ is not in the $O(1)$ neighborhood of these Peregrine and $N_0$-th order super rogue waves, then $u_N(x,t)$ asymptotically approaches the constant background $e^{\textrm{i}t}$ as $|A| \to \infty$.
\item
If $\mu\ne 0, \pm 1, \dots, \pm (N-1)$, then the rogue wave $u_N(x, t)$ asymptotically splits into $N(N+1)/2$ Peregrine waves $\hat{u}_1(x-\hat{x}_{0}, t-\hat{t}_{0}) \hspace{0.05cm} e^{\textrm{i}t}$, whose spatial-temporal locations $(\hat{x}_{0}, \hat{t}_{0})$ are given by Eq.~(\ref{x0t0}), with $z_{0}$ being each of the $N(N+1)/2$ simple nonzero roots of the Umemura polynomial $U_{N}(z;\mu)$. The error of this Peregrine wave approximation is $O(|A|^{-1})$.
\end{enumerate}
\end{thm}

This theorem gives very clear predictions on the rogue field under Umemura-type parameters (\ref{acond2}) for all $\mu$ values. The reason for these clear rogue-field predictions is that we have very clear root-multiplicity information of Umemura polynomials for all $\mu$ values, which was summarized in Lemma~\ref{Lemma1}.

The reader may notice the similarities between this theorem's results and those for Adler-Moser-type parameters in \cite{YangNLS2021,Yang2024NLS,Yang2025NLS}. In both cases, for a simple nonzero root of the associated polynomial, the  corresponding rogue field is a fundamental rogue wave far away from the origin; and for a zero (possibly repeated) root of the associated polynomial, the corresponding rogue field is a lower-order rogue wave in the $O(1)$ neighborhood of the origin. These similarities are not surprising. Indeed, these features also appear in rogue patterns associated with other polynomials such as Okamoto-hierarchy polynomials (see \cite{YangOkamoto2023}), and we have reason to believe that they are universal features of rogue wave patterns. For rogue waves associated with a given polynomial, the main work is to derive the asymptotic locations of fundamental rogue waves induced by a simple nonzero root, and the order and internal parameters of the lower-order rogue wave induced by a zero (possibly repeated) root. These tasks are sometimes easier but some other times more challenging. For example, both of these tasks are highly nontrivial for rogue patterns associated with Okamoto-hierarchy polynomials \cite{YangOkamoto2023}.

For the present rogue patterns associated with Umemura polynomials in Theorem~\ref{Theorem1}, we find that the results on fundamental (Peregrine) rogue waves induced by simple nonzero roots can be readily proved by slightly modifying the counterpart proof for Adler-Moser polynomials in \cite{Yang2024NLS}. So, that part of the proof will be omitted for brevity. We will only prove the results on the lower-order rogue wave induced by a zero root below.

\textbf{Proof.} When $\mu$ is equal to one of $0, \pm 1, \dots, \pm (N-1)$, we know from Lemma~\ref{Lemma1} that the Umemura polynomial $U_{N}(z;\mu)$ has a zero root of multiplicity $N_0(N_0+1)/2$, where $N_0=N-|\mu|$. In this case, we will show below that corresponding to this zero root, there exists a $N_0$-th order rogue wave $u_{N_0}(x, t)$ with all-zero internal parameters (thus a super rogue wave) in the $O(1)$ neighborhood of the origin, and the error of this prediction is $O(|A|^{-1})$. We will only prove this for positive values of $\mu$, i.e., $1\le \mu \le N-1$. The $\mu=0$ case is trivial, while the negative-$\mu$ case can be converted to the positive-$\mu$ case by flipping the sign of parameter $A$ in view of Eq.~(\ref{acond2}).

In the $O(1)$ neighborhood of the origin, where $x^2+t^2=O(1)$, we first rewrite the $\sigma_n$ determinant (\ref{sigma_n}) into a $3N\times 3N$ determinant \cite{OhtaJY2012}
\[ \label{sigma3Nby3N}
\sigma_{n}=\left|\begin{array}{cc}
\textbf{O}_{N\times N} & \mathbf{\Phi}_{N\times 2N} \\
-\mathbf{\Psi}_{2N\times N} & \textbf{I}_{2N\times 2N} \end{array}\right|,
\]
where
\[ \label{PhiPsi}
\Phi_{i,j}=2^{-(j-1)} S_{2i-j}\left[\textbf{\emph{x}}^{+} + (j-1) \textbf{\emph{s}}+\textbf{\emph{a}}\right], \quad
\Psi_{i,j}=2^{-(i-1)} S_{2j-i}\left[\textbf{\emph{x}}^{-} + (i-1) \textbf{\emph{s}}+\textbf{\emph{a}}^*\right].
\]
From the definition of Schur polynomials we see that
\[ \label{Sjsplit}
S_{k}(\textbf{\emph{x}}^{+} +\nu \textbf{\emph{s}}+\textbf{\emph{a}}) = \sum_{i=0}^k S_{k-i}(\textbf{\emph{a}})S_{i}(\textbf{\emph{x}}^{+} +\nu \textbf{\emph{s}}),
\]
where $\nu$ is any integer. In addition, from the definition of functions $p_{j}(z; \mu)$ in Eq.~(\ref{polypkzmu2}), we see that
\[ \label{Sjw}
S_k(\textbf{\emph{a}})=A^k p_k(0; \mu).
\]
The $p_k(0; \mu)$ value can be obtained by taking $z=0$ in Eq.~(\ref{polypkzmu2}), i.e., from the expansion
\begin{equation}\label{polypkzmu2b}
\sum_{k=0}^{\infty} p_{k}(0;\mu) \epsilon^k = \exp\left( \mu \sum_{j=1}^{\infty} \frac{1}{2j-1} \epsilon^{2j-1} \right).
\end{equation}
Then, using the above two relations (\ref{Sjsplit})-(\ref{Sjw}), we can rewrite the $\mathbf{\Phi}$ matrix in Eq.~(\ref{sigma3Nby3N}) as $\mathbf{\Phi}=\mathbf{F}\mathbf{G}$, where
\begin{eqnarray}
&& \mathbf{F}=\mbox{Mat}_{1\le i\le N,\hspace{0.06cm} 1\le j\le 2N}\left(A^{2i-j}h_{2i-j}\right)= \mathbf{D}_1\hspace{0.02cm} \mathbf{H} \hspace{0.04cm} \mathbf{D}_2,  \\
&& \mathbf{G}=\mbox{Mat}_{1\le i, j\le 2N}\left(
2^{-(j-1)} S_{i-j}\left[\textbf{\emph{x}}^{+} + (j-1) \textbf{\emph{s}}\right]\right), \\
&& \mathbf{D}_1=\mbox{diag}(A, A^3, \dots, A^{2N-1}),   \label{defD1} \\
&& \mathbf{D}_2=\mbox{diag}(1, A^{-1}, \dots, A^{1-2N}),  \label{defD2}
\end{eqnarray}
the $\mathbf{H}$ matrix is
\[ \label{defHN}
\mathbf{H}=\left(\begin{array}{ccccccccc} h_1 & 1 & & &&&    \\
                               h_3 & h_2 & h_1 & 1 &&&   \\
                               h_5 & h_4 & h_3 & h_2 & h_1 & 1 & \\
                               \vdots & \vdots & \vdots & \vdots & \vdots & \vdots & \\
                               h_{2N-1} & h_{2N-2} & h_{2N-3} &h_{2N-4} & h_{2N-5} & h_{2N-6} & \cdots & h_1 & 1
                               \end{array} \right),
\]
and $h_k \equiv p_k(0; \mu)$.

To see the structure of the above $\mathbf{H}$ matrix, it is helpful to introduce parameters $\hat{h}_k\equiv \hat{p}_k(0; \mu)$, where $\hat{p}_k(0; \mu)$ is obtained by taking $z=0$ in Eq.~(\ref{polypkzmu}), i.e., from the expansion
\begin{equation}\label{polypkzmub}
\sum_{k=0}^{\infty} \hat{p}_{k}(0;\mu) \epsilon^k =(1+ \epsilon)^\mu.
\end{equation}
Recalling that $\mu$ is a positive integer, we see that $\hat{h}_k={\displaystyle {\tbinom {\mu}{k}}}$, or more explicitly,
\[ \label{pk0mu}
\hat{h}_k=\left\{ \begin{array}{ll}
\mu (\mu-1)\cdots (\mu-k+1)/k!, & 1\le k \le \mu-1, \\
1, & k=0 \ \mbox{or} \ k=\mu,  \\
0, & k<0 \ \mbox{or} \  k>\mu.
\end{array} \right.
\]
Parameters $h_k$ in the $\mathbf{H}$ matrix (\ref{defHN}) are directly related to parameters $\hat{h}_k$. Indeed, from Eqs.~(\ref{polypkzmu}), (\ref{polypkzmu1}), (\ref{polypkzmu2}) and the above parameter definitions, we see that
\begin{equation}\label{polypkzmu2c}
\sum_{k=0}^{\infty} h_k \epsilon^k = \left(\sum_{k=0}^{\infty} \hat{h}_k \epsilon^k  \right) \exp\left( \sum_{j=1}^\infty \frac{\mu}{2j}\epsilon^{2j}\right).
\end{equation}
Thus,
\[ \label{hkhkhat}
h_k=\sum_{j=0}^{[k/2]} \gamma_j\hat{h}_{k-2j},
\]
where $\gamma_j$ are coefficients from the expansion
\[
\sum_{k=0}^\infty \gamma_k \epsilon^{2k} = \exp\left( \sum_{j=1}^\infty \frac{\mu}{2j}\epsilon^{2j}\right),
\]
and $[k/2]$ represents the largest integer less or equal to $k/2$.

Utilizing the relation (\ref{hkhkhat}) and performing elementary row operations of multiplying a row by a number and then adding it to a lower row in the $\mathbf{H}$ matrix (\ref{defHN}), we can convert this matrix to the following $\widehat{\mathbf{H}}$ matrix,
\[ \label{defHhatN}
\widehat{\mathbf{H}}=\left(\begin{array}{ccccccccc} \hat{h}_1 & 1 & & &&&    \\
                               \hat{h}_3 & \hat{h}_2 & \hat{h}_1 & 1 &&&   \\
                               \hat{h}_5 & \hat{h}_4 & \hat{h}_3 & \hat{h}_2 & \hat{h}_1 & 1 & \\
                               \vdots & \vdots & \vdots & \vdots & \vdots & \vdots & \\
                               \hat{h}_{2N-1} & \hat{h}_{2N-2} & \hat{h}_{2N-3} &\hat{h}_{2N-4} & \hat{h}_{2N-5} & \hat{h}_{2N-6} & \cdots & \hat{h}_1 & 1
                               \end{array} \right).
\]
Then, recalling the values of $\hat{h}_k$ in Eq.~(\ref{pk0mu}), we see that this  $\widehat{\mathbf{H}}$ matrix has the following structure,
\[\label{Hmatrix}
\widehat{\mathbf{H}}=\left(\begin{array}{ccccccccccc} \hat{h}_1 & 1 & \cdots & 0 &  0 & 0 &  \cdots & 0 & \cdots &  0   \\
                               \hat{h}_3 & \hat{h}_2 & \ldots & 0 &  0 & 0 &  \cdots & 0 &  \cdots &  0    \\
                               \vdots  & \vdots  & \ddots & \vdots &  \vdots  & \vdots  & \vdots & \vdots & \ddots &  \vdots  \\
                               0 & 0 & \cdots & \hat{h}_\mu &  \hat{h}_{\mu-1} & \hat{h}_{\mu-2}&  \cdots &  0 & \cdots &  0  \\
                                0 & 0 & \cdots & 0 &  0 & \hat{h}_\mu & \hat{h}_{\mu-1}&  \hat{h}_{\mu-2}& \cdots &  0  \\
                                 0 & 0 & \cdots & 0 &  0 & 0 & 0&  \hat{h}_\mu & \cdots &  0  \\
                               \vdots  & \vdots  & \ddots & \vdots &  \vdots  & \vdots  & \vdots & \vdots  & \ddots & \vdots \\
                                 0 & 0 & \cdots & 0 &  0 & 0 & 0& 0& \cdots &  1  \\
                               \end{array} \right).
\]
The determinant of every $k\times k$ submatrix (with $1\le k \le \mu$) at the top left corner of this $\widehat{\mathbf{H}}$ matrix can be seen from Eq.~(\ref{DetUmemura0}) as $U_{k}(0; \mu)/c_{k}$, which is nonzero since zero is not a root of the Umemura polynomial $U_{k}(z; \mu)$ when $\mu \ge k$ in view of Lemma~\ref{Lemma1}. This $k\times k$ determinant can also be calculated explicitly by applying the Pascal's Identity and row/column manipulations, from which one can see directly that it is nonzero. Thus, the $\mu\times \mu$ submatrix at the top left corner of this $\widehat{\mathbf{H}}$ matrix has nonzero leading principal minors, and as a consequence it can be reduced to an upper triangular matrix with nonzero diagonal elements by Gauss elimination without pivoting. Then, recalling $\hat{h}_{\mu}=1$, the row echelon form of the $\mathbf{H}$ matrix (\ref{defHN}) is then
\[ \label{defHtildeN}
\widetilde{\mathbf{H}}=\left(\begin{array}{cc} \mathbf{A} & \mathbf{C}    \\  \mathbf{0} & \mathbf{B}
\end{array} \right),
\]
where $\mathbf{A}$ is a $\mu \times \mu$ upper triangular matrix with nonzero diagonal elements, and $\mathbf{B}$ is a $N_0\times (N_0+N)$ staired matrix
\[ \label{BformLemma1}
\mathbf{B}=\left(\begin{array}{cccccccccccc} 0 & 1 &\hat{h}_{\mu-1} & \hat{h}_{\mu-2} & \dots & \dots & \dots & \dots & \dots & \dots & \dots & 0 \\
&& 0 & 1 & \hat{h}_{\mu-1} & \hat{h}_{\mu-2} & \dots & \dots & \dots & \dots & \dots & 0\\
&&&& 0 & 1 & \hat{h}_{\mu-1} & \hat{h}_{\mu-2} & \dots & \dots  & \dots & 0\\
&&&&&&\ddots & \ddots &\ddots & \ddots & \ddots & \vdots \\
&&&&&&& & 0 & 1 & \dots & 1  \end{array} \right),
\]
where $N_0=N-\mu$. Once this row echelon form $\widetilde{\mathbf{H}}$ of $\mathbf{H}$ is known, then following the steps in Ref.~\cite{Yang2025NLS}, we will see that $\sigma_n$ in Eq.~(\ref{sigma3Nby3N}) asymptotically reduces to
\[ \label{sigmasigmahat}
\sigma_n= \alpha_0 \hspace{0.04cm} \hat{\sigma}_n  \left[1+O\left(|A|^{-1}\right)\right],       \quad |A|\gg 1,
\]
where $\alpha_0$ is a certain positive constant,
\[ \label{sigmahat}
\hat{\sigma}_n=\left|\begin{array}{cc}
\textbf{O}_{N_0\times N_0} & \widehat{\mathbf{\Phi}}_{N_0\times 2N_0} \\
-\widehat{\mathbf{\Psi}}_{2N_0\times N_0} & \textbf{I}_{2N_0\times 2N_0} \end{array}\right|,
\]
and
\[ \label{PhiPsihat}
\widehat{\Phi}_{i,j}=2^{-(j-1)} S_{2i-j}\left[\textbf{\emph{x}}^{+} + (j-1+\mu) \textbf{\emph{s}}\right], \quad
\widehat{\Psi}_{i,j}=2^{-(i-1)} S_{2j-i}\left[\textbf{\emph{x}}^{-} + (i-1+\mu) \textbf{\emph{s}}\right].
\]
For the NLS equation, the odd elements $s_{2k-1}$ of the constant vector $\textbf{\emph{s}}$ are all zero, see Eq.~(\ref{sexpand}). Because of that and the index structure of the $\widehat{\mathbf{\Phi}}$ and $\widehat{\mathbf{\Psi}}$ matrices above, we can use techniques of Ref.~\cite{YangNLS2021} to remove the $\mu$ terms in the above equation (\ref{PhiPsihat}) without affecting the $\hat{\sigma}_n$ determinant. Then, the resulting $\hat{\sigma}_n$ simply corresponds to the $N_0$-th order NLS rogue wave $u_{N_0}(x, t)$ with all-zero internal parameters (thus a super rogue wave), and the error of this $u_{N_0}(x, t)$ approximation can be seen from Eqs.~(\ref{uNform}) and (\ref{sigmasigmahat}) as $O(|A|^{-1})$. This completes the proof. $\Box$

In the above proof, the key is to derive the row echelon form $\widetilde{\mathbf{H}}$ of $\mathbf{H}$ in Eq.~(\ref{defHtildeN}). As we have explained in Ref.~\cite{Yang2025NLS}, this row echelon form is equivalent to the multiplicity of the zero root in the Umemura polynomial $U_N(z; \mu)$, and one can infer one from the other. By explicitly deriving this row echelon form above, we basically provided an alternative proof for the multiplicity of the zero root in the $U_N(z; \mu)$ polynomial given in Lemma~\ref{Lemma1}.

\subsection{Numerical confirmation of Umemura-type rogue patterns}
Now, we use two examples to confirm Theorem~\ref{Theorem1}.

In our first example, we take
\[ \label{para3}
N=5, \quad A=6, \quad \mu=2
\]
in the rogue wave solution (\ref{uNform}) with Umemura-type parameters (\ref{acond2}). The exact solution $|u_5(x, t)|$ is plotted in the upper-left panel of Fig.~\ref{figNLS}. This solution contains nine widely-separated lumps on two concentric arcs, with five and four lumps on the outer and inner arc, respectively. Each lump has a peak amplitude of 3 approximately, and it is flanked by two dips on the two sides of the $x$-direction. A quick examination shows that each lump is an approximate Peregrine wave. At the center of the two arcs lies a more delicate wave structure with a central hump of amplitude 7 approximately, flanked by lower peaks and valleys nearby. Examination of this wave structure shows that it is an approximate third-order super rogue wave
\cite{AAS2009, Clarkson2010, He2017, Miller2020, Yangbook2024}.

Next, we use Theorem~\ref{Theorem1} to predict the wave field of this solution. The root structure of the underlying Umemura polynomial $U_5(z; 2)$ in Fig.~\ref{fig1}(b) comprises 9 simple nonzero roots lying on two concentric arcs and a zero root of multiplicity 6. Then, according to Theorem~\ref{Theorem1}, the solution $u_5(x, t)$ would asymptotically split into 9 Peregrine waves whose locations are given by Eq.~(\ref{x0t0}), plus a third-order super rogue wave. This predicted solution is plotted in the upper right panel of Fig.~\ref{figNLS}. By comparing this predicted solution with the true solution in the upper left panel, we can clearly see that the predicted solution matches the true solution very well.

In our second example, we take
\[ \label{para4}
N=5, \quad A=15, \quad \mu=1/100
\]
in the rogue wave solution (\ref{uNform}) with Umemura-type parameters (\ref{acond2}). The exact solution $|u_5(x, t)|$ is plotted in the lower left panel of Fig.~\ref{figNLS}. We can see that this solution comprises 9 approximate Peregrine waves on an outer ring, another 5 approximate Peregrine waves on an inner ring, plus one more approximate Peregrine wave at the center of the two rings. A wave pattern similar to this has been numerically reported in \cite{He2013} (see Fig.~6 there).

Analytically, we can predict this wave pattern using Theorem~\ref{Theorem1}. The root structure of the underlying Umemura polynomial $U_5(z; \mu)$ with $\mu=1/100$ in Fig.~\ref{fig1}(d) comprises 9 simple roots on an outer ring, another 5 simple roots on an inner ring, plus one more simple root at the center of the two rings (this center root is approximately $-0.002844$, which is very small but nonzero). Then, Theorem~\ref{Theorem1} predicts that the wave field of the solution $u_5(x, t)$ would split into 15 Peregrine waves whose locations are given by Eq.~(\ref{x0t0}). This predicted solution is plotted in the lower right panel of Fig.~\ref{figNLS}. Comparison of this predicted solution to the true solution shows that they agree with each other very well.

Quantitatively, we have also measured the errors of our asymptotic predictions versus the large parameter $|A|$ for both examples in Fig.~\ref{figNLS}. That error analysis confirmed that the errors do decay in proportion to $|A|^{-1}$ as Theorem~\ref{Theorem1} predicted. Details of this quantitative comparison are omitted here for brevity.

\begin{figure}[htb]
\begin{center}
\includegraphics[scale=0.35, bb=350 000 650 1050]{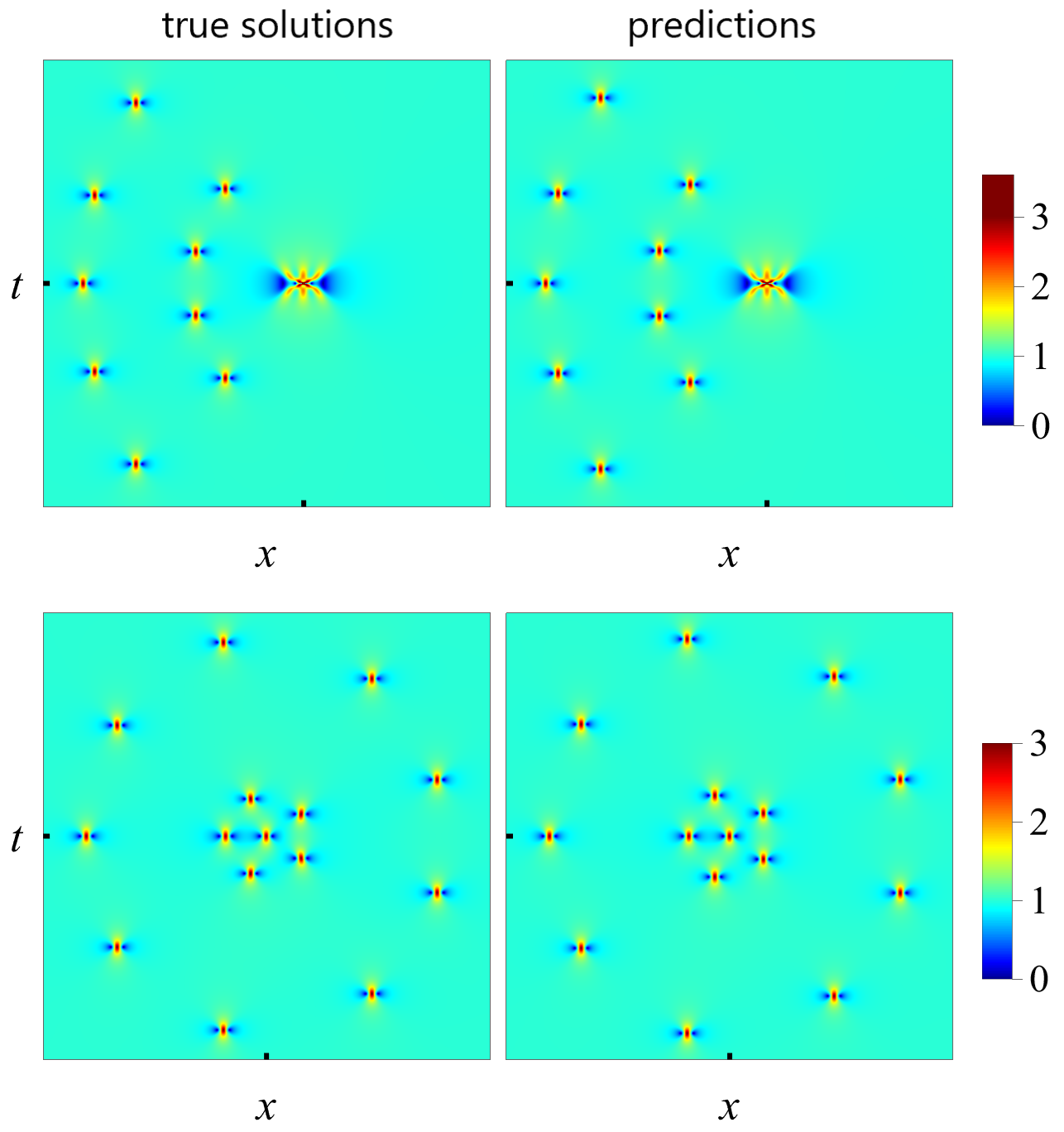}
\caption{ Comparisons between true NLS rogue solutions $|u_5(x,t)|$ (left panels) and their analytical predictions (right panels) under Umemura-type internal parameters (\ref{acond2}). Upper row: $A=6,\ \mu=2$; lower row: $A=15,\ \mu=1/100$. In the upper row,  $ -35 \leq x \leq 25,\  -30 \leq t  \leq 30$; in the lower row, $ -30 \leq x, t  \leq 30$. For the panels in the upper row, the color bar does not show the full solution range.  }  \label{figNLS}
\end{center}
\end{figure}

\subsection{Universality of Umemura-type rogue patterns}
The link between rogue patterns and Umemura polynomials is certainly not restricted to the NLS equation. Those Umemura-type rogue patterns will reliably appear in many other integrable equations such as the Boussinesq equation, the derivative NLS equations, and the Manakov system, when internal parameters in their rogue solutions are of Umemura-type (\ref{acond2}), similar to the universality of rogue patterns associated with Yablonskii-Vorob'ev hierarchy polynomials that was reported earlier in \cite{Yang2021b,Yangbook2024}. We will briefly demonstrate this Umemura-type rogue universality on the Boussinesq equation below.

The Boussinesq equation was introduced in 1871 for the propagation of long surface waves on water of constant depth \cite{JBoussinesq1871,JBoussinesq1872}. After variable normalizations and variable shifts, this equation can be written as \cite{Clarkson_Boussi}
\[ \label{BoussiEq}
u_{tt}+u_{xx}-(u^2)_{xx}-\frac{1}{3}u_{xxxx}=0.
\]
Special rogue wave solutions in this equation were derived in \cite{Tajiri1,Rao2017,Clarkson_Boussi}.
General rogue waves in this equation were derived in \cite{BoussiRWs2020} and then simplified in \cite{Yang2021b}. Expressions of these simplified general rogue waves are \cite{Yang2021b}
\[  \label{Boussi_rogue}
u_{N}(x,t)=2 \partial_{x}^2 \ln \sigma ,
\]
where
\[ \label{SigmanAlg}
\sigma(x,t)=\det_{
\begin{subarray}{l}
1\leq i, j \leq N
\end{subarray}
}
\left(  \phi_{2i-1,2j-1} \right),
\]
\[ \label{matrixmij}
\phi_{i,j}=\sum_{\nu=0}^{\min(i,j)} \left(\frac{-1}{12}\right)^{\nu} \hspace{0.06cm} S_{i-\nu}(\textbf{\emph{x}}^{+} +\nu \textbf{\emph{s}}+\textbf{\emph{a}})  \hspace{0.06cm} S_{j-\nu}(\textbf{\emph{x}}^{-} + \nu \textbf{\emph{s}}-\textbf{\emph{a}}^*),
\]
vectors $\textbf{\emph{x}}^{\pm}=\left(  x_{1}^{\pm}, 0, x_{3}^{\pm}, 0, \cdots \right)$ are defined by
\[
x_{2j+1}^{+}=\frac{\sqrt{3}\hspace{0.04cm}\textrm{i}}{2\cdot 3^{2j+1}\cdot (2j+1)!} \hspace{0.03cm} \left( x + 2^{2j} \textrm{i}t \right), \quad x_{2j+1}^{-}=-\left(x_{2j+1}^{+}\right)^*,  \quad j=0, 1, 2, \cdots,
\]
$\emph{\textbf{s}}=(s_{1}, s_{2}, \cdots)$ are coefficients from the expansion
\[
\sum_{j=1}^{\infty} s_{j}\lambda^{j}=\ln\left[ \frac{2\textrm{i}\sqrt{3}}{\lambda} \tanh\frac{\lambda}{6}\tanh\left(\frac{\lambda}{6}+ \frac{2\textrm{i}\pi}{3}\right) \right], \label{skrkexpcoeffBoussi}
\]
$\textbf{\emph{a}}=(a_1, 0, a_3, 0, \dots, a_{2N-1})$ is the vector of internal parameters, and $a_{1}, a_{3}, a_{5}, \cdots, a_{2N-1}$ are free complex constants. The first few $s_k$ values from the expansion (\ref{skrkexpcoeffBoussi}) are
\begin{equation} \label{skB}
s_1=\frac{2 \rm{i}}{3 \sqrt{3}}, \quad s_2=-\frac{5}{108}, \quad s_3=-\frac{5 \rm{i}}{243 \sqrt{3}}.
\end{equation}

The fundamental Boussinesq rogue wave $u_1(x, t)$ is obtained by setting $N=1$ and $a_1=0$ in the above solution, and we get
\[ \label{u1Bou}
u_1(x, t)=2 \partial_{x}^2  \ln  \left(x^2+t^2+1\right)=
\frac{4 \left(t^2-x^2+1\right)}{\left(t^2+x^2+1\right)^2}.
\]
This rogue wave has a single hump of amplitude 4, flanked by two dips on the two sides of the $x$ direction. For $u_N(x, t)$ with higher $N$ values, more delicate wave patterns can appear depending on its internal parameter values \cite{Clarkson_Boussi, BoussiRWs2020,Yang2021b}.

In these rogue solutions $u_N(x, t)$ in Eq.~(\ref{Boussi_rogue}), when we take their internal parameters $(a_1, a_3, \dots, a_{2N-1})$ as Umemura-type parameters (\ref{acond2}) with large $|A|$, we can use an analysis similar to the NLS one in Sec.~\ref{secTheorem1} to obtain the following result.

\begin{thm} \label{Theorem2}
For the Boussinesq rogue wave $u_N(x, t)$ in Eq.~(\ref{Boussi_rogue}) under Umemura-type parameters (\ref{acond2}), asymptotic predictions for the rogue field at large $|A|$ are as follows.
\begin{enumerate}
\item
If $\mu$ is equal to one of $0, \pm 1, \dots, \pm (N-1)$, then the rogue wave $u_N(x, t)$ asymptotically splits into $N_p$ fundamental rogue waves located far away from the spatial-temporal origin, plus a $N_0$-th order rogue wave in the $O(1)$ neighborhood of the origin, where $N_p$ and $N_0$ are as given in Lemma~\ref{Lemma1}. These fundamental rogue waves are $u_1(x-\hat{x}_{0}, t-\hat{t}_{0})$, where $u_1(x, t)$ is given in Eq.~(\ref{u1Bou}), and their spatial-temporal locations $(\hat{x}_{0}, \hat{t}_{0})$ are given by
\[
\hat{x}_{0}+\textrm{i}\hspace{0.05cm}\hat{t}_{0}=\left(-2\rm{i} \sqrt{3}\right)z_{0}A-\frac{4}{3} (N-1), \label{x0t0B}
\]
with $z_{0}$ being each of the $N_p$ simple nonzero roots of the Umemura polynomial $U_{N}(z;\mu)$. The $N_0$-th order rogue wave $u_{N_0}(x, t)$ is given by Eq.~(\ref{Boussi_rogue}), with its internal parameters $(\hat{a}_1, \hat{a}_3, \cdots, \hat{a}_{2N_0-1})$ given as
\[  \label{akBoussi}
\hat{a}_{2k-1}=|\mu| \hspace{0.04cm} s_{2k-1}, \quad k=1, 2, \cdots, N_0,
\]
where $s_{2k-1}$ are defined by the expansion (\ref{skrkexpcoeffBoussi}). The errors of these fundamental rogue wave and $N_0$-th order rogue wave approximations are $O(|A|^{-1})$.
\item
If $\mu\ne 0, \pm 1, \dots, \pm (N-1)$, then the rogue wave $u_N(x, t)$ asymptotically splits into $N(N+1)/2$ fundamental rogue waves $u_1(x-\hat{x}_{0}, t-\hat{t}_{0})$, whose spatial-temporal locations $(\hat{x}_{0}, \hat{t}_{0})$ are given by Eq.~(\ref{x0t0B}), with $z_{0}$ being each of the $N(N+1)/2$ simple nonzero roots of the Umemura polynomial $U_{N}(z;\mu)$. The error of this fundamental rogue wave approximation is $O(|A|^{-1})$.
\end{enumerate}
\end{thm}

We can see that these results for the Boussinesq equation are quite similar to those for the NLS equation in Theorem~\ref{Theorem1}. In both cases, the rogue patterns are closely linked to the roots of Umemura polynomials. The main difference between them is probably the internal-parameter values in the rogue wave  $u_{N_0}(x, t)$ near the origin. In the Boussinesq case, these internal parameters are no longer zero, see Eq.~(\ref{akBoussi}). This difference is caused by the fact that the $\mu$ terms in the reduced determinant $\hat{\sigma}_n$ in Eq.~(\ref{PhiPsihat}) cannot be dropped here since the odd elements $s_{2k-1}$ of the vector $\textbf{\emph{s}}$ are nonzero for the Boussinesq equation. Another difference between these results and those of NLS in Theorem~\ref{Theorem1} is that there is an extra term of $4/3(N-1)$ in the prediction (\ref{x0t0B}) for the spatial-temporal locations $(\hat{x}_{0}, \hat{t}_{0})$ of fundamental rogue waves. This extra term is caused by the fact that the $s_1$ coefficient in the Boussinesq case is nonzero (see Eq.~(\ref{skB})), and why this nonzero $s_1$ would induce that extra term was explained in \cite{Yang2021b}.

It is important for us to point out that, in the Boussinesq case, the predicted rogue wave $u_{N_0}(x, t)$ near the origin, with internal parameters in Eq.~(\ref{akBoussi}), is generally not a super rogue wave (i.e., the rogue wave with the highest possible peak amplitude among rogue waves of that order). This is obvious, because for a fixed $N_0$ value, where $N_0=N-|\mu|$, this predicted $u_{N_0}(x, t)$ wave can have different internal parameters (\ref{akBoussi}) due to different $|\mu|$ values, and it is impossible for these $u_{N_0}(x, t)$ waves at different internal parameters to be all super rogue waves. For example, at $(N, \mu)$ values of $(3, 1)$, $(4, 2)$ and $(5,3)$, where $N_0=2$, we get three predicted $u_2(x, t)$ waves in the neighborhood of the origin with internal parameter values $(\hat{a}_1, \hat{a}_3)$ as $(s_1, s_3)$, $(2s_1, 2s_3)$ and $(3s_1, 3s_3)$, respectively. The peak amplitudes of these three predicted $u_2(x, t)$ waves are approximately 5.4326, 5.4995 and 5.4743, while the peak amplitude of the second-order super Boussinesq rogue wave is 5.5 \cite{BoussiRWs2020}. Thus, none of these three predicted $u_2(x, t)$ waves is a super rogue wave. This fact contrasts the NLS case in Theorem~\ref{Theorem1}, where the predicted rogue wave $u_{N_0}(x, t)$ near the origin is always a super rogue wave.

Now, we use two examples to confirm the asymptotic predictions in Theorem~\ref{Theorem2}. In these examples, we take rogue wave solutions $u_5(x, t)$ with Umemura-type parameters (\ref{acond2}), where $(\mu, A)=(2, 4)$ and $(1/100, 10)$, respectively. For these two sets of parameters, the exact Boussinesq rogue waves $u_5(x, t)$ from Eq.~(\ref{Boussi_rogue}) are plotted in the left panels of Fig.~\ref{figBou}. In the upper left panel, the wave field comprises 9 approximate fundamental rogue waves on two concentric arcs, plus a more delicate wave structure at the arc center which, upon closer examination, proves to be an approximate third-order rogue wave \cite{BoussiRWs2020,Yang2021b}. In the lower left panel, the wave field comprises 15 approximate fundamental rogue waves located on two concentric rings and the ring center. In both panels, one can notice a close connection between the wavefield patterns and root structures of the associated Umemura polynomials $U_{5}(z;\mu)$ in Fig.~\ref{fig1}(b, d). In addition, these wavefield patterns closely resemble those of NLS rogue waves in Fig.~\ref{figNLS}. For comparison, the predicted rogue solutions from Theorem~\ref{Theorem2} are plotted in the right panels of Fig.~\ref{figBou}. One can see that the predicted patterns closely match the true ones.

\begin{figure}[htb]
\begin{center}
\includegraphics[scale=0.35, bb=350 000 650 1050]{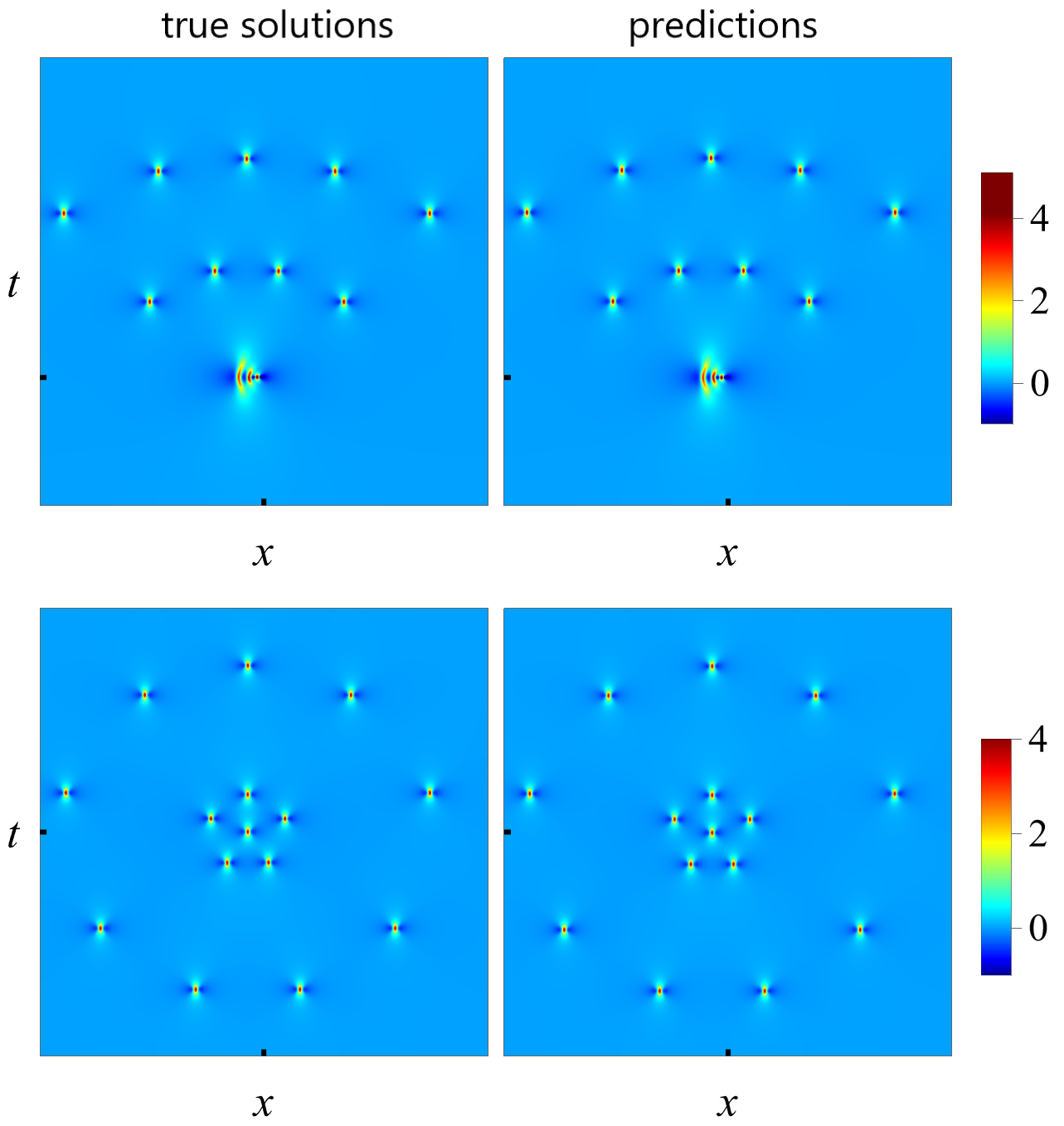}
\caption{ Comparisons between true Boussinesq rogue solutions $u_5(x,t)$ (left panels) and their analytical predictions (right panels) under Umemura-type internal parameters (\ref{acond2}).  Upper row: $A=4,\ \mu=2$; lower row: $A=10,\ \mu=1/100$. In the upper row,  $ -70 \leq x \leq 70,\  -40 \leq t  \leq 100$; in the lower row, $ -75 \leq x, t  \leq 75$.  }  \label{figBou}
\end{center}
\end{figure}

\section{Lump patterns associated with Umemura polynomials in the KPI equation} \label{sec:KP}

The Kadomtsev-Petviashvili-I (KPI) equation is
\[\label{KPI}
\left( u_t+ 6 u u_x +u_{xxx} \right)_x - 3u_{yy} =0.
\]
This equation arises in many branches of physics, such as plasma waves \cite{KP1970}, water waves \cite{Ablowitz1979}, nonlinear optics \cite{Pelinovsky1995,Baronio2016,KPIexperiment} and Bose-Einstein condensates \cite{BarashenkovBEC,Tsuchiya_BEC2008,Panos2018}. It is solvable by the inverse scattering transform \cite{Zakharov_book, Ablowitz_book}.

\subsection{A special class of higher-order lump solutions}
The KPI equation admits many types of solutions, such as fundamental lumps \cite{Petviashvili1976,Manakov1977,Ablowitz_Satsuma1979}, multi-lumps that describe the elastic interactions of fundamental lumps with different asymptotic velocities \cite{Manakov1977,Ablowitz_Satsuma1979}, higher-order lumps that describe anomalous scatterings of fundamental lumps with the same asymptotic velocities \cite{Peli93b,Ablowitz97,Ablowitz2000}, and others \cite{He_KPI_rogue_lump,Zakharov2021}. In this article, we are concerned with higher-order lumps (also called multi-pole lumps in the literature \cite{Ablowitz97,Ablowitz2000}).  These solutions are interesting because interactions of their constituent fundamental lumps are not elastic at all, and patterns of the constituent fundamental lumps can change drastically after interactions.

Higher-order lumps have been derived and studied in many articles before \cite{Peli93b,Ablowitz97,Ablowitz2000,DGKM2010,
Chen2016,Clarkson_Boussi,Gaillard2018,Chang2018,Guo2022,YangYangKPI, YangYangKPI2}. After parameter normalizations due to invariances of the equation, the simplest expressions of these higher-order lumps are \cite{YangYangKPI, YangYangKPI2}
\[ \label{Schpolysolu0}
u_{N}(x,y,t)=2 \partial_{x}^2 \ln \sigma,
\]
where
\[\label{Blockmatrix0}
\sigma(x,y,t)=\det_{1 \leq i,j \leq N}\left(\phi_{ij}\right),
\]
$N$ is a positive integer,
\[ \label{Schmatrimnij0}
\phi_{i, j}=\sum_{\nu=0}^{\min(n_i,n_j)} \frac{1}{4^{\nu}} \hspace{0.06cm} S_{n_i-\nu}\left(\textbf{\emph{x}}^{+} +\nu \textbf{\emph{s}}+\textbf{\emph{a}}_{i} \right)  \hspace{0.06cm} S_{n_j-\nu}\left((\textbf{\emph{x}}^{+})^* + \nu \textbf{\emph{s}}+ \textbf{\emph{a}}_{j}^*\right),
\]
$\Lambda \equiv (n_{1}, n_{2}, \cdots n_N)$ is an order-index vector of arbitrary distinct positive integers, the vector $\textbf{\emph{x}}^{+}=\left( x_{1}^{+}, x_{2}^{+},\cdots \right)$ is defined as
\[
x_{k}^{+}= \frac{x+2^k\textrm{i} y-3^k(4t)}{k!},   \quad k\ge 1, \label{defxrp}
\]
the vector $\textbf{\emph{s}}=(s_1, s_2, \cdots)$ is as given in Eq.~(\ref{sexpand}), vectors $\textbf{\emph{a}}_{i}$ are internal parameters with $\textbf{\emph{a}}_{i}=\left( a_{i,1}, a_{i,2}, \cdots , a_{i,n_i}\right)$, and $a_{i, j} \hspace{0.05cm} (1\le i\le  N,  1\le j\le n_i)$ are free complex constants.

Large-time patterns of these lump solutions were determined in \cite{YangYangKPI, YangYangKPI2}, and those patterns were found to be associated with Yablonskii-Vorob'ev polynomials, Wronskian-Hermite polynomials and certain two-term polynomials (these two-term polynomials are special members of the Yablonskii-Vorob'ev hierarchy polynomials). To obtain lump patterns associated with Umemura polynomials, we need to restrict to a special class of the above solutions.

First, we need to choose the order-index vector $\Lambda$ as $\Lambda=(1, 3, 5, \dots, 2N-1)$. We also need to choose internal parameters $\textbf{\emph{a}}_{i}$ so that $a_{i, j}$ is independent of the $i$ index. In this case, each $\textbf{\emph{a}}_{i}$ vector with $i<N$ is just a truncation of the longest vector $\textbf{\emph{a}}_{N}$. Since every $\textbf{\emph{a}}_{i}$ can be extended to the full $\textbf{\emph{a}}_{N}$ (or even longer), and the extended parts are dummy parameters which do not appear in the actual solution formulae, by performing this $\textbf{\emph{a}}_{i}$ extension, we can say all $\{\textbf{\emph{a}}_{i}\}$ vectors are the same in this case and thus denote $\textbf{\emph{a}}_{i}=\textbf{\emph{a}}$, where $\textbf{\emph{a}}=(a_1, a_2, a_3, \dots, a_{2N-1})$.

Under the above $\Lambda$ and $\textbf{\emph{a}}_{i}$ choices, we can use the technique outlined in Appendix A of Ref.~\cite{YangNLS2021} to eliminate all $x_{even}^+$ and $a_{even}$ terms from the vectors $\textbf{\emph{x}}^{+}$ and $\textbf{\emph{a}}$ without affecting the $u_{N}(x,y,t)$ solution. The reduced expressions of this special class of higher-order lump solutions are given in the following lemma.

\begin{lem}  \label{Lemma4}
The KPI equation (\ref{KPI}) admits a special class of higher-order lumps whose expressions are
\[ \label{Schpolysolu}
u_N(x,y,t)=2 \partial_{x}^2 \ln \sigma,
\]
where $N$ is a positive integer representing the order of the solution,
\[\label{Blockmatrix}
\sigma(x,y,t)=\det_{\begin{subarray}{l}
1\leq i, j \leq N
\end{subarray}}
\left( \phi_{2i-1,2j-1} \right),
\]
\[ \label{Schmatrimnij}
\phi_{i, j}=\sum_{\nu=0}^{\min(i,j)} \frac{1}{4^{\nu}} \hspace{0.06cm} S_{i-\nu}\left(\textbf{\emph{x}}^{+} +\nu \textbf{\emph{s}}+\textbf{\emph{a}} \right)  \hspace{0.06cm} S_{j-\nu}\left((\textbf{\emph{x}}^{+})^* + \nu \textbf{\emph{s}}+ \textbf{\emph{a}}^*\right),
\]
the vector ${\textbf{\emph{x}}}^{+} \equiv \left(x_1^+, 0, x_3^+, 0, \cdots\right)$ is defined by
\[ \label{defxkKP}
x_{2j-1}^{+}= \frac{x+2^{2j-1} \textrm{i} y-3^{2j-1}(4t)}{(2j-1)!}, \quad j\ge 1,
\]
the vector $\textbf{\emph{s}}=(0, s_2, 0, s_4, \cdots)$ is as given in Eq.~(\ref{sexpand}), $\textbf{\emph{a}}=(a_1, 0, a_3, 0, \dots, a_{2N-1})$ is the vector of internal parameters, and $a_{1}, a_{3}, a_{5}, \cdots, a_{2N-1}$ are free complex constants.
\end{lem}

We will be concerned with this special class of higher-order lumps in the rest of this section.

The fundamental lump of the KPI equation is obtained when we choose $N=1$ and $a_1=0$ in the above class of  lump solutions, and its expression is
\[ \label{defu1}
u_1(x,y,t)=2 \partial_{x}^2 \ln \left(\left(x-12t  \right)^2+4y^2+ \frac{1}{4}\right)=
\frac{1-4(x-12 t)^2+16 y^2}{\left((x-12t)^2+4 y^2+\frac{1}{4}\right)^2}.
\]
This is a moving single lump with peak amplitude $16$, which is attained at the spatial location of $(x, y)=(12t, \hspace{0.05cm} 0)$. Very recently, this fundamental lump was observed in nonlinear photorefractive crystals for the first time \cite{KPIexperiment}.

\subsection{Connections between KPI lump solutions and NLS rogue wave solutions}

By comparing this special class of higher-order lumps of KPI in Lemma~\ref{Lemma4} with the general rogue wave solutions of NLS in Lemma~\ref{Lemma2}, we can see that their expressions are very similar. In fact, we can see immediately that the $\sigma_0(x, t)$ function of NLS in Eq.~(\ref{sigma_n}) becomes the above $\sigma(x, y, t)$ function of KPI in Eq.~(\ref{Blockmatrix}) under the following variable substitutions:
\[ \label{transform}
t \hspace{0.02cm} \to \hspace{0.02cm} 2y, \quad a_{2j-1} \hspace{0.02cm} \to \hspace{0.02cm} a_{2j-1}-\frac{3^{2j-1}(4t)}{(2j-1)!}, \hspace{0.15cm} 1\le j\le N.
\]
That is,
\[ \label{NLSKP0}
\sigma^{[KP]}(x, y, t; a_1, a_3, \dots, a_{2N-1})=\sigma_0^{[NLS]}(x, t; a_1, a_3, \dots, a_{2N-1})|_{t \hspace{0.02cm} \to \hspace{0.02cm} 2y, \hspace{0.05cm} a_{2j-1} \hspace{0.02cm} \to \hspace{0.02cm} a_{2j-1}-3^{2j-1}(4t)/(2j-1)!}.
\]
In addition, under the above variable substitutions, the KPI's lump solution $u_N(x, y, t)$ in Eq.~(\ref{Schpolysolu}) can be obtained directly from NLS's rogue wave solution $u_N(x, t)$ in Eq.~(\ref{uNform}) as
\[ \label{NLSKP}
u_N^{[KP]}(x, y, t; a_1, a_3, \dots, a_{2N-1})=\left. 2\left(|u_N^{[NLS]}(x, t; a_1, a_3, \dots, a_{2N-1})|^2-1\right)\right|_{t \hspace{0.02cm} \to \hspace{0.02cm} 2y, \hspace{0.05cm} a_{2j-1} \hspace{0.02cm} \to \hspace{0.02cm} a_{2j-1}-3^{2j-1}(4t)/(2j-1)!}.
\]
The reason is that the $\sigma_0$ and $\sigma_1$ functions in the NLS rogue solution (\ref{uNform}) satisfy the bilinear equation
\[
(D_x^2 + 2) \hspace{0.04cm} \sigma_0 \cdot \sigma_0 = 2|\sigma_1|^2,
\]
where $D$ is Hirota's bilinear differential operator, see Ref.~\cite{OhtaJY2012}. Utilizing this bilinear equation, one can easily check that $2(|\sigma_1/\sigma_0|^2-1)=2 \partial_{x}^2 \ln \sigma_0$, which leads to Eq.~(\ref{NLSKP}) in view of Eqs.~(\ref{uNform}), (\ref{Schpolysolu}) and (\ref{NLSKP0}).

From the above KPI-NLS connection (\ref{NLSKP}), we can immediately see that the highest possible peak amplitude of the special KPI lump solutions $u_N(x, y, t)$ in Lemma~\ref{Lemma4} is $2[(2N+1)^2-1]$, i.e., $8N(N+1)$, because the highest possible peak amplitude of NLS rogue waves $u_N(x, t)$ is believed to be $2N+1$ (see \cite{AAS2009, Clarkson2010, He2017, Miller2020, Yangbook2024}). In addition, the lump solution $u_N(x, y, t)$ in Lemma~\ref{Lemma4} with all-zero internal parameters (i.e., $a_1=a_3=\dots=a_{2N-1}=0$) is a super lump, i.e., it reaches that highest possible peak amplitude $8N(N+1)$, at $x=y=t=0$. When time becomes large positive or negative, this super lump splits into $N(N+1)/2$ fundamental lumps \cite{YangYangKPI}. Notice that the amplitude of each fundamental lump is 16, and $16$ times $N(N+1)/2$ is $8N(N+1)$, which is the amplitude of this super lump at $t=0$, this super lump at $t=0$ then can be viewed as $N(N+1)/2$ fundamental lumps coherently added together.

We should point out that the connection between KPI's lump solutions and NLS' rogue solutions has been reported before in Ref.~\cite{DGKM2010}. However, the variable substitutions from NLS rogue waves to KPI lumps in \cite{DGKM2010} are very different from ours in Eq.~(\ref{transform}), because the solution form and parameterizations of NLS rogue waves in \cite{DGKM2010} are very different from ours in Lemma~\ref{Lemma2}. In \cite{DGKM2010}, the $N$-th order NLS rogue waves $u_N(x, t)$ were expressed by determinants with matrix elements involving differential operators, and they were parameterized by $2N$ real constants $\varphi_{j}$ ($1\le j\le 2N$). The variable substitutions from their NLS rogue solutions to KPI lumps were $t\to y$ and $\varphi_3\to t$, which is in stark contrast with our variable substitutions (\ref{transform}), primarily in the mapping of the parameters.

\subsection{Lump patterns under Umemura-type internal parameters}
Lump patterns associated with Umemura polynomials will arise in the special class of higher-order lump solutions of Lemma~\ref{Lemma4} in the parameter regime
\[ \label{acondKP}
a_{2j-1}=\frac{\mu}{2j-1} \hspace{0.04cm} A^{2j-1}, \quad 1\le j\le N, \quad
t=O(1),
\]
where $\mu$ is a free $O(1)$ complex constant, and $A$ is another free complex constant with large modulus, i.e., $|A|\gg 1$. In this parameter regime, we have the following theorem.

\begin{thm} \label{Theorem3}
For the $N$-th order KPI lump solution $u_N(x, y, t)$ in Lemma~\ref{Lemma4} under Umemura-type parameters (\ref{acondKP}), asymptotic predictions of the solution field at large $|A|$ are as follows.
\begin{enumerate}
\item
If $\mu$ is equal to one of $0, \pm 1, \dots, \pm (N-1)$, then the wave field asymptotically splits into $N_p$ fundamental lumps located far away from the lump center $(x_0, y_0)=(12t, 0)$, plus a $N_0$-th order super lump in the $O(1)$ neighborhood of this lump center, where $N_p$ and $N_0$ are as given in Lemma~\ref{Lemma1}. These fundamental lumps are $u_{1}(x-\hat{x}_{0}, \hspace{0.04cm}  y-\hat{y}_{0}, t)$, where $u_1(x,y,t)$ is given in Eq.~(\ref{defu1}), and their positions $(\hat{x}_{0}, \hat{y}_{0})$ are given by
\[\label{x0y0}
\hat{x}_0+2\textrm{i}\hat{y}_0= z_0 A,
\]
with $z_{0}$ being each of the $N_p$ simple nonzero roots of the Umemura polynomial $U_{N}(z;\mu)$. The peak of each fundamental lump is spatially located at $(x, y)=(12 t+\hat{x}_0, \hat{y}_0)$. The $N_0$-th order super lump $u_{N_0}(x, y, t)$ is given by Eqs.~(\ref{Schpolysolu})-(\ref{defxkKP}) with all its internal parameters $(\hat{a}_1, \hat{a}_3, \cdots, \hat{a}_{2N_0-1})$ being zero. The errors of these fundamental-lump and $N_0$-th order super lump approximations are $O(|A|^{-1})$. If $(x,y)$ is not in the $O(1)$ neighborhood of these fundamental lumps and $N_0$-th order super lump, then $u_N(x,y,t)$ asymptotically approaches zero as $|A| \to \infty$.
\item
If $\mu\ne 0, \pm 1, \dots, \pm (N-1)$, then the wave field asymptotically splits into $N(N+1)/2$ fundamental lumps $u_1(x-\hat{x}_{0}, \hspace{0.04cm}  y-\hat{y}_{0}, t)$, where $(\hat{x}_{0}, \hat{y}_{0})$ are given by Eq.~(\ref{x0y0}), with $z_{0}$ being each of the $N(N+1)/2$ simple nonzero roots of the Umemura polynomial $U_{N}(z;\mu)$. The error of this fundamental-lump approximation is $O(|A|^{-1})$.
\end{enumerate}
\end{thm}

{\bf Proof.} We first rewrite the determinant of $\sigma(x, y, t)$ in Eq.~(\ref{Blockmatrix}) as a larger $3N \times 3N$ determinant \cite{OhtaJY2012}
\[ \label{3Nby3Ndet2}
\sigma=\left|\begin{array}{cc}
\textbf{O}_{N\times N} & \Phi_{N\times 2N} \\
-\Psi_{2N\times N} & \textbf{I}_{2N \times 2N} \end{array}\right|,
\]
where
\[
\Phi_{i,j}=2^{-(j-1)} S_{2i-j}\left(\textbf{\emph{x}}^{+} + (j-1) \textbf{\emph{s}} +\textbf{\emph{a}}\right), \quad \Psi_{i,j}=2^{-(i-1)} S_{2j-i}\left((\textbf{\emph{x}}^{+})^* + (i-1) \textbf{\emph{s}}+\textbf{\emph{a}}^*\right).
\]
Performing the Laplace expansion to this larger determinant, we get
\[ \label{sigmanLap}
\sigma=\sum_{1\leq\nu_{1} < \nu_{2} < \cdots < \nu_{N}\leq 2N}\left|
\det_{1 \leq i, j\leq N} \Phi_{i, \nu_j}\right|^2.
\]

We first prove that when $(x-\hat{x}_{0})^2+(y-\hat{y}_{0})^2=O(1)$, where $(\hat{x}_{0}, \hat{y}_{0})$ are given by Eq.~(\ref{x0y0}), then the higher-order lump solution $u_N(x, y, t)$ reduces to a fundamental lump $u_{1}(x-\hat{x}_{0}, \hspace{0.04cm}  y-\hat{y}_{0}, t)$ in this neighborhood. In this proof, we will use the Laplace-expansion form (\ref{sigmanLap}) of the $\sigma$ function. To proceed, we notice that under Umemura-type parameters (\ref{acondKP}) we have the asymptotics
\begin{eqnarray}
&& S_k\left(\textbf{\emph{x}}^{+} + \nu \textbf{\emph{s}} +\textbf{\emph{a}}\right)=S_k\left(x+2\textrm{i}y-12t+a_1, 0, x/6+4\textrm{i}y/3-18t+a_3, 0, \dots\right)  \nonumber \\
&& = S_k\left(x+2\textrm{i}y-12t+a_1, 0, a_3, 0, a_5, 0, \dots\right)\left(1+O(|A|^{-2})\right)  \nonumber \\
&& = S_k\left(\hat{x}_0+2\textrm{i}\hat{y}_0+a_1+\hat{x}+2\textrm{i}\hat{y}-12t, 0, a_3, 0, a_5, 0, \dots\right)\left(1+O(|A|^{-2})\right)  \nonumber \\
&& = A^k S_k\left[ z_0+ \mu + A^{-1}\left(\hat{x}+2\textrm{i}\hat{y}-12t\right), 0, \mu/3, 0, \mu/5, 0, \dots\right]\left(1+O(|A|^{-2})\right)  \nonumber \\
&& = A^k p_k\left(z_0+A^{-1}(\hat{x}+2\textrm{i}\hat{y}-12t); \mu\right)\left(1+O(|A|^{-2})\right),   \label{SkasymKP}
\end{eqnarray}
where $\hat{x}\equiv x-\hat{x}_0$, $\hat{y}\equiv y-\hat{y}_0$, and $p_k(z; \mu)$ are Schur polynomials defined in Eq.~(\ref{polypkzmu2}). At large $|A|$, the dominant contributions to the Laplace expansion of $\sigma$ in Eq.~(\ref{sigmanLap}) come from two index choices, $\nu =(0, 1, . . . , N-1)$ and $(0, 1, . . . , N-2, N)$. For the first index choice, using the above asymptotics (\ref{SkasymKP}) we get
\[
\det_{1 \leq i, j\leq N} \Phi_{i, \nu_j}=\alpha \hspace{0.06cm} A^{\frac{N(N+1)}{2}}U_{N}\left(z_0+A^{-1}\left(\hat{x}+2\textrm{i}\hat{y}-12t\right); \mu \right) \left[1+O\left(|A|^{-2}\right)\right],
\]
where $\alpha=2^{-N(N-1)/2}c_N^{-1}$. Since $z_0$ is a simple root of the Umemura polynomial $U_N(z; \mu)$, performing Taylor expansion to the above $U_N$ function, we get
\[
\det_{1 \leq i, j\leq N} \Phi_{i, \nu_j}=\alpha \hspace{0.06cm} A^{\frac{N(N+1)}{2}-1}U'_{N}(z_0; \mu) \left(\hat{x}+2\textrm{i}\hat{y}-12t\right)\left[1+O\left(|A|^{-1}\right)\right],
\]
where the prime represents differentiation to $z$, and $U'_{N}(z_0; \mu)\ne 0$. For the second index choice of $\nu =(0, 1, . . . , N-2, N)$, using the asymptotics (\ref{SkasymKP}) we get
\[
\det_{1 \leq i, j\leq N} \Phi_{i, \nu_j}=\frac{1}{2}\alpha \hspace{0.06cm}
A^{\frac{N(N+1)}{2}-1}U'_{N}\left(z_0; \mu \right) \left[1+O\left(|A|^{-1}\right)\right].
\]
Collecting these two dominant contributions, we get from Eq.~(\ref{sigmanLap}) that
\[\label{sigma2}
\sigma(x,y,t) = \alpha^2 \hspace{0.06cm} \left|U_{N}'(z_0; \mu)\right|^2 |A|^{N(N+1)-2} \left[ \left(x-\hat{x}_{0}-12t\right)^2+4\left(y-\hat{y}_{0}\right)^2 +\frac{1}{4} \right]  \left[1+O\left(|A|^{-1}\right)\right].
\]
Plugging this asymptotics into Eq.~(\ref{Schpolysolu}), we then see that this gives a fundamental lump $u_{1}(x-\hat{x}_{0}, \hspace{0.04cm}  y-\hat{y}_{0}, t)$, and the error of this prediction is $O(|A|^{-1})$.

Next, we prove that if $U(z; \mu)$ has a zero root of multiplicity $N_0(N_0+1)/2$, then when $t=O(1)$, the higher-order lump solution $u_N(x, y, t)$ reduces to a lower-order lump solution $u_{N_0}(x, y, t)$ with all-zero internal parameters in the $O(1)$ neighborhood of the spatial origin $(x, y)=(0, 0)$. In this proof, we will use the $3N\times 3N$ determinant form (\ref{3Nby3Ndet2}) of the $\sigma$ function. To do so, we notice that when $x^2+y^2=O(1)$ and $t=O(1)$, we have
\[ \label{SjsplitKP}
S_{k}(\textbf{\emph{x}}^{+} +\nu \textbf{\emph{s}}+\textbf{\emph{a}}) = \sum_{i=0}^k S_{k-i}(\textbf{\emph{a}})S_{i}(\textbf{\emph{x}}^{+} +\nu \textbf{\emph{s}}),
\]
which is the same relation as (\ref{Sjsplit}) in the NLS case except that the contents of the $\textbf{\emph{x}}^{+}$ vector are different here. The $\textbf{\emph{a}}$ vector here is the same as that in the NLS case. Thus, we also have $S_k(\textbf{\emph{a}})=A^k p_k(0; \mu)$ as in Eq.~(\ref{Sjw}) before.
Using these relations, we can rewrite the $\mathbf{\Phi}$ matrix in Eq.~(\ref{3Nby3Ndet2}) as $\mathbf{\Phi}=\mathbf{F}\mathbf{G}$, where
\begin{eqnarray}
&& \mathbf{F}=\mbox{Mat}_{1\le i\le N,\hspace{0.06cm} 1\le j\le 2N}\left(A^{2i-j}h_{2i-j}\right)= \mathbf{D}_1\hspace{0.02cm} \mathbf{H} \hspace{0.04cm} \mathbf{D}_2,  \\
&& \mathbf{G}=\mbox{Mat}_{1\le i, j\le 2N}\left(
2^{-(j-1)} S_{i-j}\left[\textbf{\emph{x}}^{+} + (j-1) \textbf{\emph{s}}\right]\right),
\end{eqnarray}
and the $\mathbf{D}_1, \mathbf{D}_2, \mathbf{H}$ matrices are as given in Eqs.~(\ref{defD1})-(\ref{defHN}). We have shown earlier in Sec.~\ref{secTheorem1} that the row echelon form of the $\mathbf{H}$ matrix is of the form given in Eq.~(\ref{defHtildeN}). Using this row echelon form and following the steps in Ref.~\cite{Yang2025NLS}, we will see that $\sigma(x, y, t)$ in Eq.~(\ref{3Nby3Ndet2}) asymptotically reduces to
\[ \label{sigmasigmahat2}
\sigma \sim \alpha_0 \hspace{0.04cm} \hat{\sigma}  \left[1+O\left(|A|^{-1}\right)\right],       \quad |A|\gg 1,
\]
where $\alpha_0$ is a certain positive constant,
\[ \label{sigmahat2}
\hat{\sigma}=\left|\begin{array}{cc}
\textbf{O}_{N_0\times N_0} & \widehat{\mathbf{\Phi}}_{N_0\times 2N_0} \\
-\widehat{\mathbf{\Psi}}_{2N_0\times N_0} & \textbf{I}_{2N_0\times 2N_0} \end{array}\right|,
\]
and
\[ \label{PhiPsihat2}
\widehat{\Phi}_{i,j}=2^{-(j-1)} S_{2i-j}\left[\textbf{\emph{x}}^{+} + (j-1+\mu) \textbf{\emph{s}}\right], \quad
\widehat{\Psi}_{i,j}=2^{-(i-1)} S_{2j-i}\left[(\textbf{\emph{x}}^{+})^* + (i-1+\mu) \textbf{\emph{s}}\right].
\]
Since the odd elements of the constant vector $\textbf{\emph{s}}$ are all zero in view of Eq.~(\ref{sexpand}), we can use techniques of Ref.~\cite{YangNLS2021} to remove the $\mu$ terms in the above equation (\ref{PhiPsihat2}) without affecting the $\hat{\sigma}$ determinant. Then, the resulting $\hat{\sigma}$  corresponds to the $N_0$-th order KPI lump solution $u_{N_0}(x,y,t)$ given by Eqs.~(\ref{Schpolysolu})-(\ref{defxkKP}) with all its internal parameters $(\hat{a}_1, \hat{a}_3, \cdots, \hat{a}_{2N_0-1})$ being zero, and the error of this $u_{N_0}(x, y, t)$ approximation can be seen from Eqs.~(\ref{Schpolysolu}) and (\ref{sigmasigmahat2}) as $O(|A|^{-1})$. This completes the proof of Theorem~\ref{Theorem3}. $\Box$

The above proofs on Umemura-related lump patterns are very similar to those on rogue patterns in the previous section. The reason is very simple. In the special higher-order lumps of Lemma~\ref{Lemma4}, their $\sigma$ function is very similar to the $\sigma_n$ function of the NLS equation in Lemma~\ref{Lemma2}. Because of that, it is not surprising that our treatments for the two cases are very similar.

\subsection{Numerical confirmation of Umemura-type lump patterns}

Now, we use two examples to numerically confirm the predictions of Theorem~\ref{Theorem3} on Umemura-type lump patterns. We take Umemura-type parameters (\ref{acondKP}) in the fifth-order lump solutions $u_5(x, y, t)$ of Lemma~\ref{Lemma4} with $(A, \mu)$ as $(10, 2)$ in the first example and $(20, 1/100)$ in the second example. Our predictions of their wavefields at $t=-1, 0$ and 1 from Theorem~\ref{Theorem3} are plotted in Fig.~\ref{figKP1}. These time values are chosen since Theorem~\ref{Theorem3} requires time to be $O(1)$, see Eq.~(\ref{acondKP}).

In the upper panels of Fig.~\ref{figKP1} for the first example, the predicted wavefields comprise 9 fundamental lumps on two concentric arcs, plus a third-order super lump at the arc center. These 9 fundamental lumps move at the same speed 12 toward the positive $x$-axis. Thus, in the moving frame $\hat{x}\equiv x-12t$ as used in these panels, these arcs are stationary. The third-order super lump at the arc center has a complex profile with extreme height of 96 at $t=0$ (see the middle panel). When $|t|$ moves away from 0, this wave complex splits into 6 approximate fundamental lumps in a triangular pattern (see the left and right panels). This triangle expands as $|t|$ increases, and its pointing directions at positive and negative times are opposite of each other (pointing rightward at negative times and leftward at positive times).

In the lower panels of Fig.~\ref{figKP1} for the second example, the predicted wavefields comprise 15 fundamental lumps located on two concentric rings and the ring center. These 15 fundamental lumps move at the same speed 12 toward the positive $x$-axis. Thus, in the moving frame $\hat{x}\equiv x-12t$ as used in these panels, these patterns of concentric rings with a ring center are stationary.

\begin{figure}[htb]
\begin{center}
\includegraphics[scale=0.48, bb=400 000 500 350]{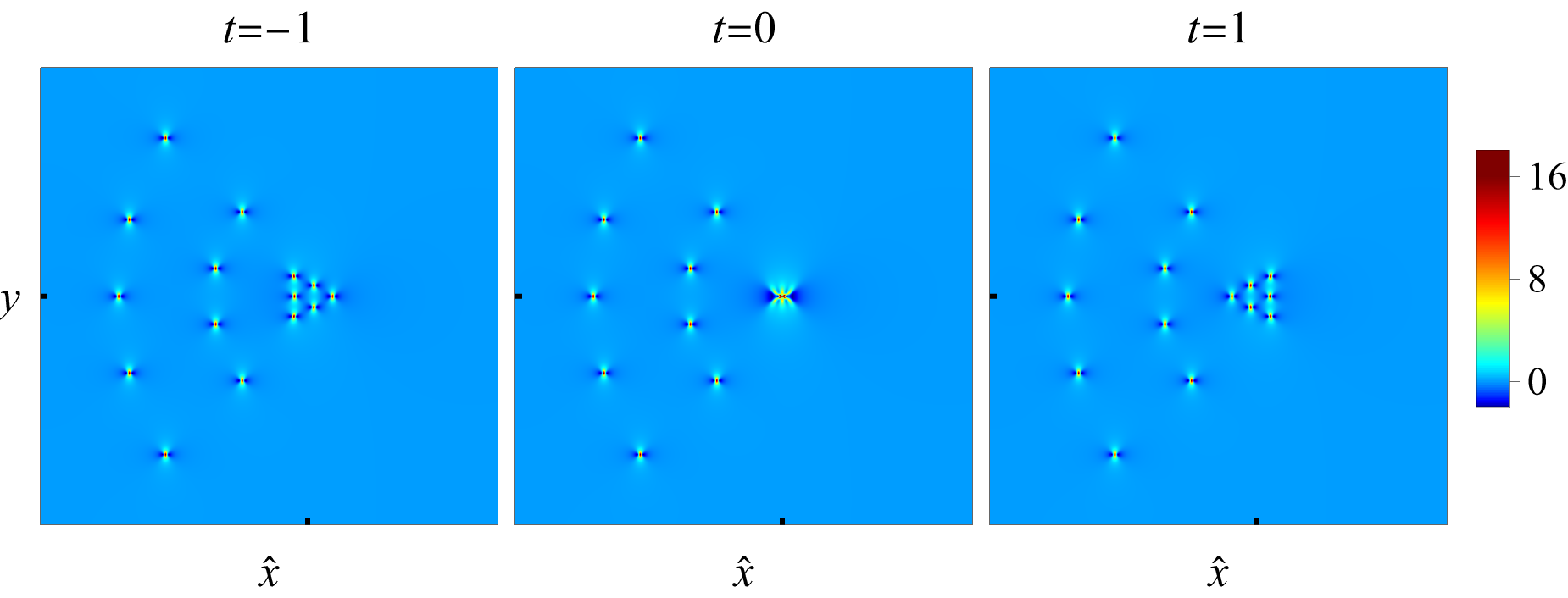}

\includegraphics[scale=0.48, bb=400 000 500 365]{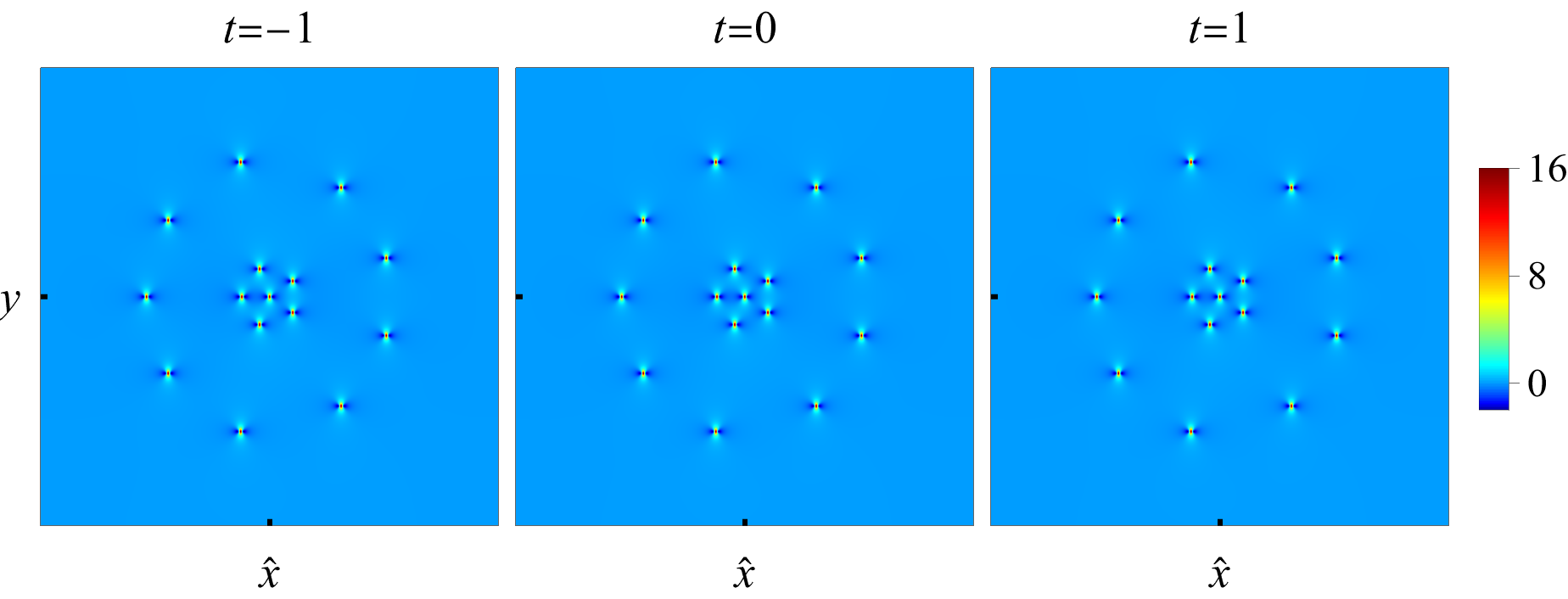}
\caption{Predicted KPI lump solutions $u_5(x, y, t)$ under Umemura-type parameters (\ref{acondKP}) at three time values of $t=-1, 0$ and $1$.  Upper row: $A=10$,\ $\mu=2$; lower row: $A=20$,\ $\mu=1/100$. The $(\hat{x}, y)$ intervals are $-70\leq \hat{x} \leq 50,\  -30 \leq y  \leq 30$ for the upper row,  and $-60 \leq \hat{x} \leq 60,\  -30 \leq y  \leq 30$ for the lower row, where $\hat{x}\equiv x-12t$ is the moving $x$-coordinate. For the upper middle panel,  the color bar does not show the full solution range. } \label{figKP1}
\end{center}
\end{figure}

Now, we verify these predictions. For the first example with $N=5, \mu=2$ and $A=10$, the true solutions $u_5(x, y, t)$ at six time values of $-100, -1, 0, 1, 50$ and 200 are displayed in Fig.~\ref{figKP2}. We can see that the wave pattern has gone through huge changes over these times. In the $t=-100$ panel, the wave field comprises 15 (approximate) fundamental lumps arranged in a triangular pattern pointing rightward. As time increases, the nine fundamental lumps vertically aligned on the far left of the wavefield bend over rightward and become two concentric arcs, while the remaining six fundamental lumps on a smaller right-pointing triangle continuously move toward each other (see the $t=-1$ panel). At $t\approx 0$, these six converging fundamental lumps coalesce, forming a complex lump structure with extreme height of approximately 95.19 (see the $t=0$ panel). Afterwards, that lump complex separates into a triangle of six fundamental lumps again but with the $x$-direction reversed --- pointing toward the left now (see the $t=1$ panel). This triangle of six fundamental lumps continues to expand, and when they meet the two arcs of 9 fundamental lumps on their left, a complex series of anomalous scatterings take place which throw the constituent fundamental lumps to new positions (see the $t=50$ panel). Eventually, these 15 fundamental lumps evolve into a big triangle again which points leftward, see the $t=200$ panel.

\begin{figure}[htb]
\begin{center}
\includegraphics[scale=0.48, bb=400 000 500 350]{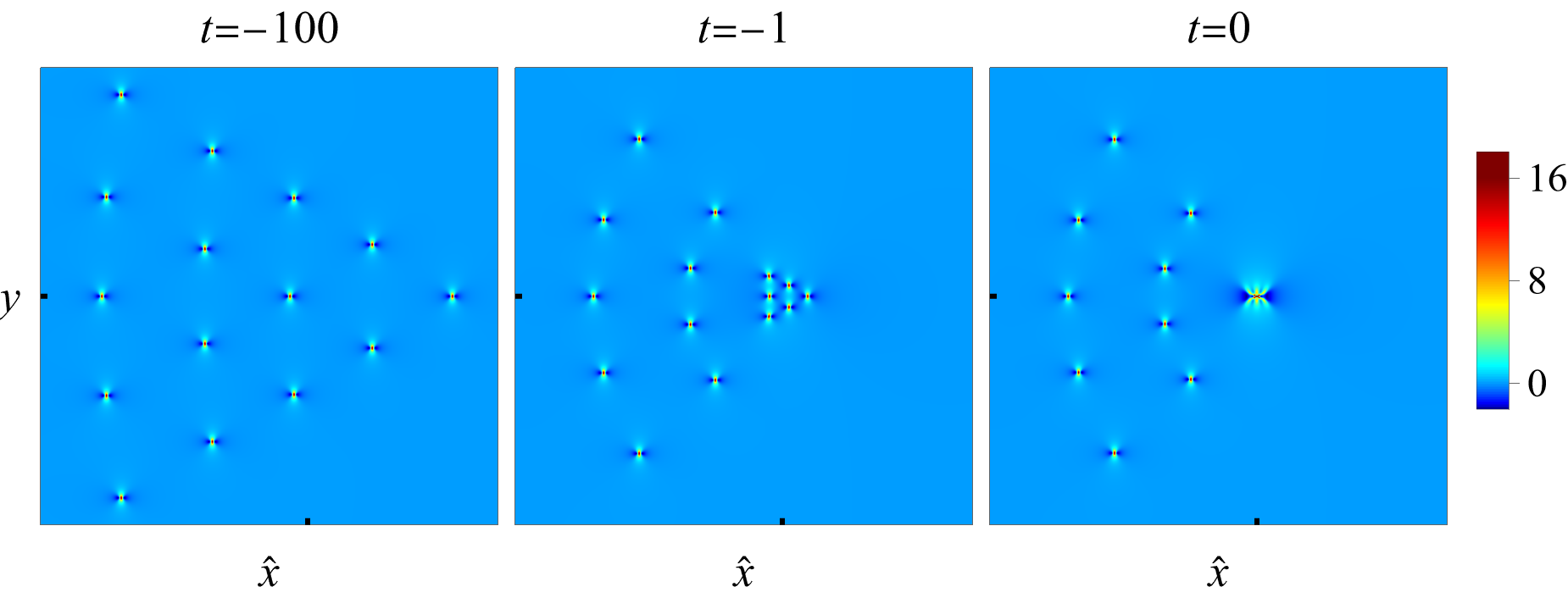}

\includegraphics[scale=0.48, bb=400 000 500 365]{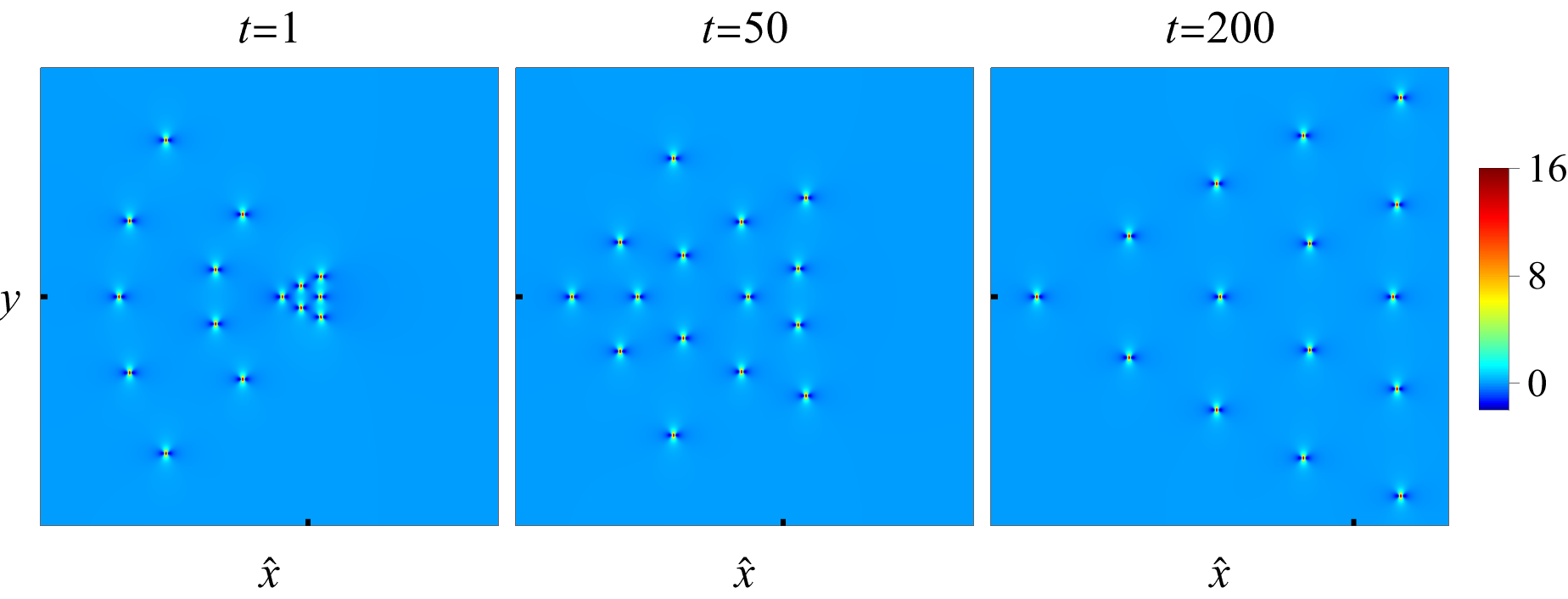}
\caption{True KPI lump solutions $u_5(x, y, t)$ under Umemura-type parameters (\ref{acondKP}) with $A=10$ and $\mu=2$ at six time values of $t=-100, -1, 0, 1, 50$ and 200. The $(\hat{x}, y)$ intervals are $-70\leq \hat{x} \leq 50,\  -30 \leq y  \leq 30$ for $t=-100, -1, 0, 1, 50$, and $-95\leq \hat{x} \leq 25,\  -30 \leq y  \leq 30$ for $t=200$, where $\hat{x}\equiv x-12t$ is the moving $x$-coordinate. These true solutions at $t=-1, 0$ and 1 are to be compared with their predictions in the upper row of Fig.~\ref{figKP1}. For the upper right panel,  the color bar does not show the full solution range.  } \label{figKP2}
\end{center}
\end{figure}

The big triangular patterns at large negative and positive times of $t=-100$ and 200 in Fig.~\ref{figKP2} have been predicted in our earlier work \cite{YangYangKPI}. Our focus here is the panels of $t=-1, 0$ and $1$ that fall into the parameter regimes of Theorem~\ref{Theorem3} and were predicted in the upper row of Fig.~\ref{figKP1}. By comparing these true solutions at $t=-1, 0$ and $1$ with those predicted ones in the upper row of Fig.~\ref{figKP1}, we can see that they closely match each other, confirming the validity of Theorem~\ref{Theorem3}. We have further conducted an error analysis for our predictions by measuring the errors of fundamental-lump predictions on the two arcs and the third-order super-lump prediction at the arc center versus the asymptotic parameter $|A|$, and this error analysis confirmed that the error indeed decays in proportion to $|A|^{-1}$ as Theorem~\ref{Theorem3} predicted. Details of this error analysis are omitted here for brevity.

\begin{figure}[htb]
\begin{center}
\includegraphics[scale=0.48, bb=400 000 500 350]{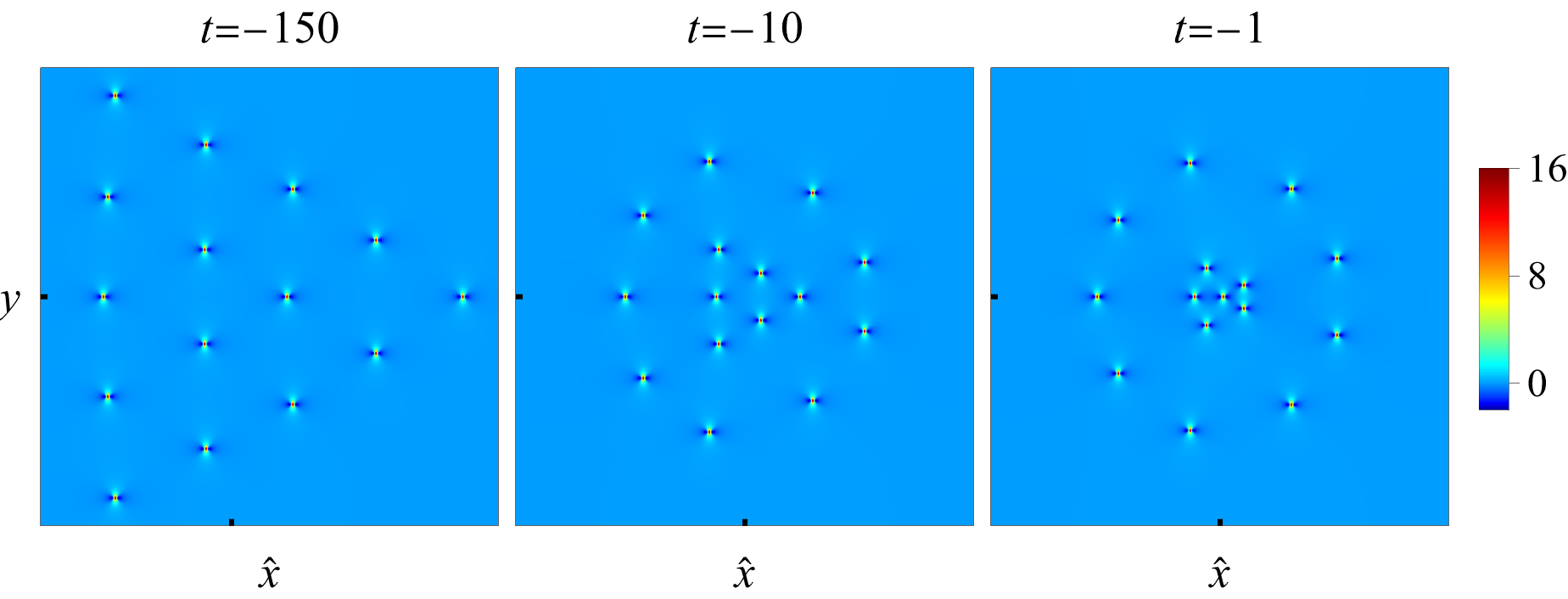}

\includegraphics[scale=0.48, bb=400 000 500 365]{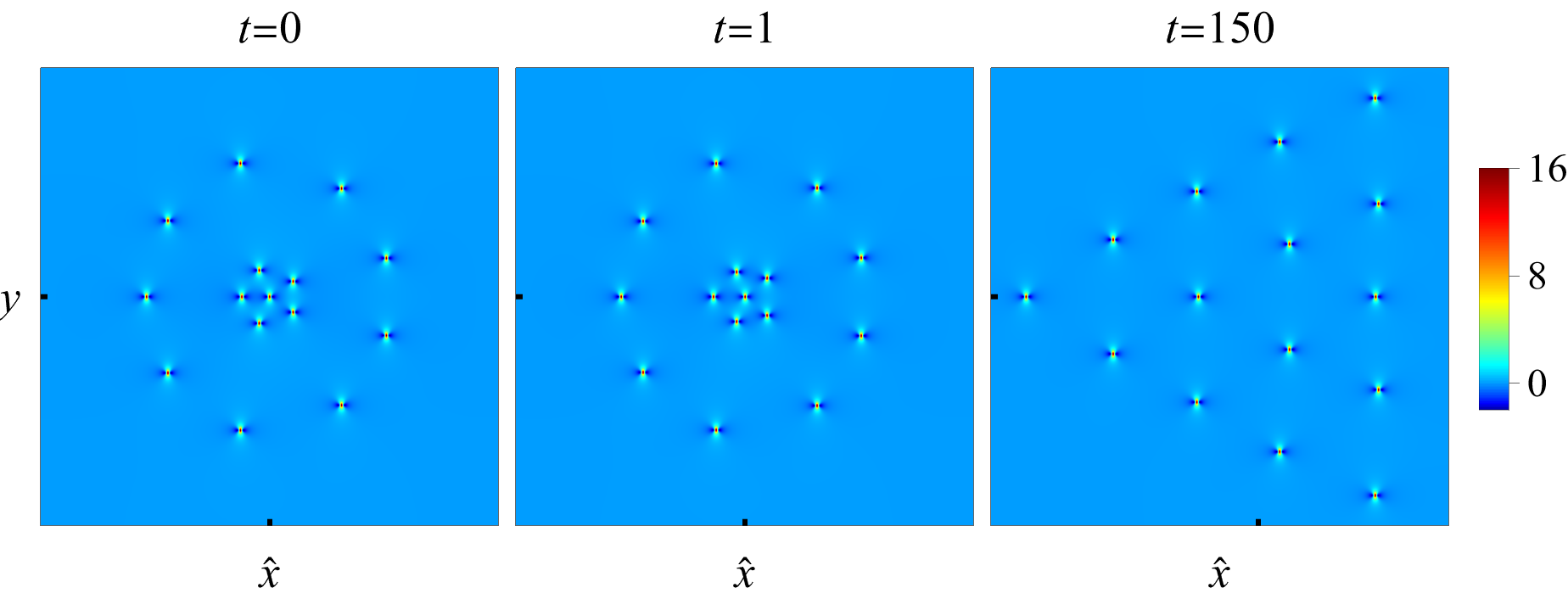}
\caption{True KPI lump solutions $u_5(x, y, t)$ under Umemura-type parameters (\ref{acondKP}) with $A=20$ and $\mu=1/100$ at six time values of $t=-150, -10, -1, 0, 1$ and $150$. The $(\hat{x}, y)$ intervals are $-50 \leq \hat{x} \leq 70,\  -30 \leq y  \leq 30$ for $t=-150$, $-60 \leq \hat{x} \leq 60,\  -30 \leq y  \leq 30$ for $t=-10, -1, 0, 1$, and $-70 \leq \hat{x} \leq 50,\  -30 \leq y  \leq 30$  for $t=150$, where $\hat{x}\equiv x-12t$ is the moving $x$-coordinate.  These true solutions at $t=-1, 0$ and 1 are to be compared with their predictions in the lower row of Fig.~\ref{figKP1}.   } \label{figKP3}
\end{center}
\end{figure}

\begin{figure}[htb]
\begin{center}
\includegraphics[scale=0.48, bb=200 000 500 380]{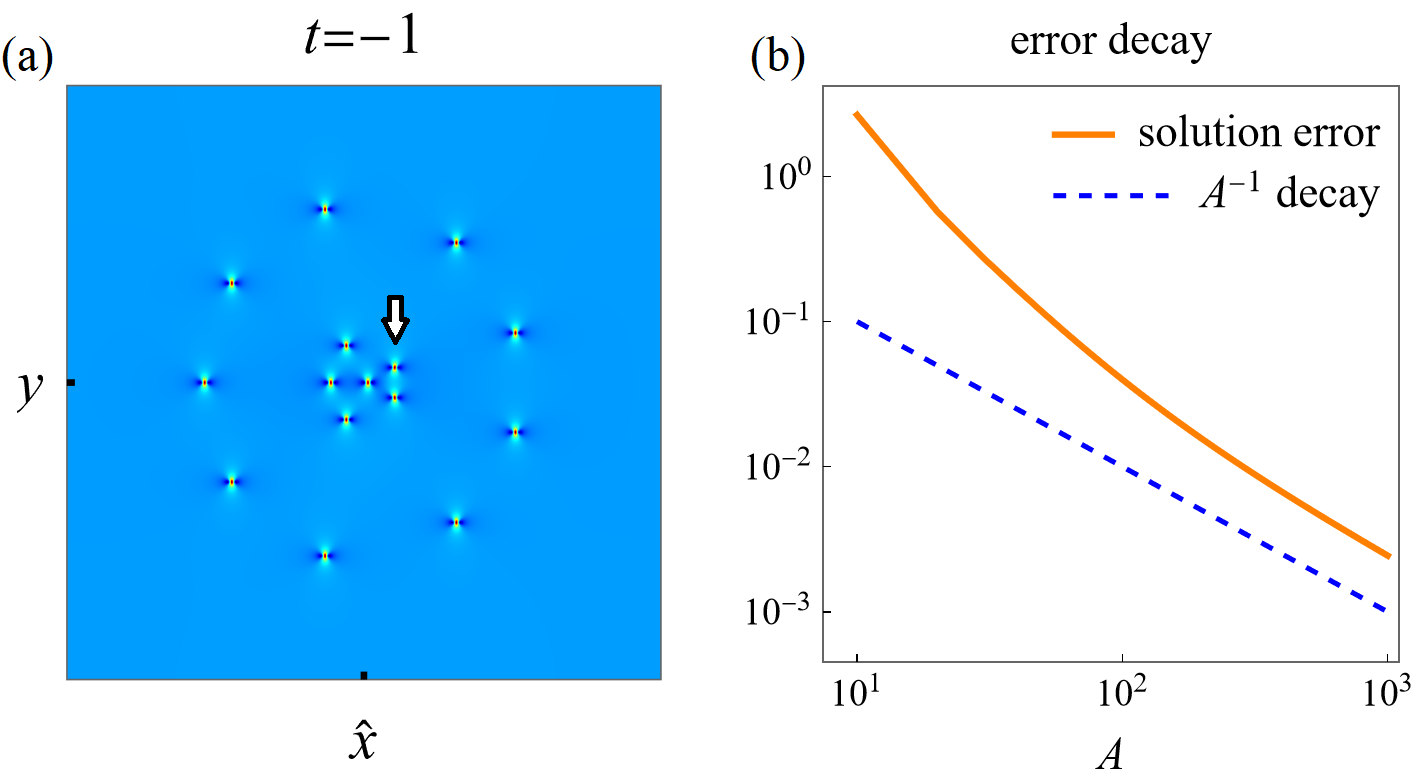}
\caption{Quantitative confirmation of Theorem~\ref{Theorem3} for the KPI lump solution in Fig.~\ref{figKP3} at $t=-1$. (a) Graph of this true solution at $A=20$. (b) Error diagram for the location of one of the fundamental lumps marked by a white arrow in panel (a) versus the $A$ value.  The predicted decay rate of $O(A^{-1})$ is also plotted for comparison. } \label{figKP4}
\end{center}
\end{figure}

To confirm our predictions for the second example in the lower row of Fig.~\ref{figKP1}, we plot in Fig.~\ref{figKP3} true KPI lump solutions $u_5(x, y, t)$ under Umemura-type parameters (\ref{acondKP}) with $A=20$ and $\mu=1/100$ at six time values of $t=-150, -10, -1, 0, 1$ and 150. In the $t=-150$ panel, the wave field comprises 15 (approximate) fundamental lumps arranged in a triangular pattern pointing rightward, similar to the $t=-100$ panel of the first example in Fig.~\ref{figKP2}. However, later evolution of this triangular pattern here is very different from that in Fig.~\ref{figKP2} earlier. In this case, as time increases, the 15 fundamental lumps on the triangle engage in a series of interactions and anomalous scatterings, which turn them into a very different configuration of a ring of 9 fundamental lumps at the outside and a right-pointing smaller triangle of 6 fundamental lumps inside, see the $t=-10$ panel. As time increases further, the outer ring of 9 fundamental lumps stays roughly the same, but the inner triangle of 6 fundamental lumps goes through delicate interactions and turns into a pentagon (a smaller ring) with a center, see the $t=-1$ panel. This configuration of two concentric rings with a center stays roughly the same in the $t=-1, 0$ and 1 panels. As time continues to increase, the 6 fundamental lumps in the inner ring and the center evolve into a small triangle again but pointing leftward. Later on, the outer ring of 9 fundamental lumps starts to interact with this inner triangle of 6 fundamental lumps and rearrange into a big triangle again but pointing leftward now, see the $t=150$ panel.

The big triangular patterns at large negative and positive times of $t=\pm 150$ in Fig.~\ref{figKP3} have been predicted in our earlier work \cite{YangYangKPI}. The panels of $t=-1, 0$ and $1$ are in the parameter regimes of Theorem~\ref{Theorem3} and were predicted in the lower row of Fig.~\ref{figKP1}. By visually comparing true solutions at $t=-1, 0$ and $1$ in Fig.~\ref{figKP3} with the predicted ones in Fig.~\ref{figKP1}, we see again that they match each other quite well, especially at $t=0$ and 1. At $t=-1$, the true solution shows some visible difference from the predicted one in the lower row of Fig.~\ref{figKP1} (left panel). For example, the pentagon near the center of the $t=-1$ panel in the true solution is a little deformed compared to that in the predicted solution. In order to check the quantitative accuracy of our Theorem~\ref{Theorem3}, we have performed an error analysis for this $t=-1$ solution with varying values of the asymptotic parameter $A$. Specifically, we pick one of the fundamental lumps in that pentagon marked by an arrow in Fig.~\ref{figKP4}(a), and record the error of its peak location in our predicted solution for $A$ values ranging from 10 to 1000. This error is measured as the distance from the true location of the peak to its predicted location, i.e., $\sqrt{(x_{pre}-x_{true})^2+(y_{pre}-y_{true})^2}$. The graph of this error versus the $A$ value is shown in Fig.~\ref{figKP4}(b). According to Theorem~\ref{Theorem3}, this error is predicted to be of $O(|A|^{-1})$ for $|A|\gg 1$. Fig.~\ref{figKP4}(b) shows that this error indeed decays at the rate of $A^{-1}$ at large $A$, which quantitatively confirms the validity of Theorem~\ref{Theorem3}. This error analysis also indicates that the visual difference between the deformed pentagon of the true solution and the more regular pentagon of the predicted solution at $t=-1$ is due to our asymptotic parameter $A=20$ for that solution not large enough. If $|A|$ gets larger, that deformed pentagon will approach the predicted more regular pentagon.

\section{Conclusion and discussions} \label{sec:conclusion}
In this paper, we have studied rogue-wave and lump patterns associated with Umemura polynomials in rational solutions of the third Painlev\'{e} equation $\mbox{P}_{\mbox{\scriptsize III}}$. We have shown that in many integrable equations such as the NLS equation and the Boussinesq equation, when internal parameters of their rogue wave solutions are large and of Umemura-type (\ref{acond2}), then their rogue patterns in the spatial-temporal plane would be asymptotically predicted by root structures of Umemura polynomials, where every simple root of the Umemura polynomial would induce a fundamental rogue wave whose spatial-temporal location is linearly related to that simple root, while a multiple root of the Umemura polynomial would induce a non-fundamental rogue wave in the $O(1)$ neighborhood of the spatial-temporal origin (see Theorems~\ref{Theorem1} and \ref{Theorem2}). We have also shown that in a certain class of higher-order lump solutions of the KPI equation, when their internal parameters are large and of Umemura-type (\ref{acondKP}), then their lump patterns at $t=O(1)$ would also be predicted asymptotically by root structures of Umemura polynomials, where simple and multiple roots of the polynomial would give rise to fundamental and non-fundamental lumps in the spatial plane, respectively (see Theorem~\ref{Theorem3}). We have compared all our analytical predictions of rogue-wave and lump patterns to true solutions and demonstrated good agreement between them. These results linked the third Painlev\'{e} equation $\mbox{P}_{\mbox{\scriptsize III}}$ to nonlinear wave patterns for the first time (to our best knowledge). In these studies of NLS rogue waves and KPI lumps, we have also discovered a new variable transformation (\ref{NLSKP}) that converts our NLS rogue waves of bilinear form to higher-order lumps of KPI.

It is noted that Umemura polynomials arise in rational solutions of the third Painlev\'{e} equation (\ref{PainleveIII}) in the generic case of $\gamma \delta \ne 0$. In the degenerate case of $\gamma \delta = 0$, the $\mbox{P}_{\mbox{\scriptsize III}}$ equation admits another class of algebraic solutions which are expressed through Ohyama polynomials \cite{Clarkson2003PIII, Ohyama2006, MillerD7, Harrow2025}. An interesting open question is whether there exist rogue-wave and lump patterns that can be predicted by those Ohyama polynomials of the degenerate $\mbox{P}_{\mbox{\scriptsize III}}$. This question lies outside the scope of this paper and will be left for future studies.

In this paper, we reported nonlinear wave patterns that are associated with the third Painlev\'{e} equation $\mbox{P}_{\mbox{\scriptsize III}}$. As we have reviewed in the introduction, nonlinear wave patterns associated with $\mbox{P}_{\mbox{\scriptsize I}}$, $\mbox{P}_{\mbox{\scriptsize II}}$ and $\mbox{P}_{\mbox{\scriptsize IV}}$ have been reported earlier as well. Then, what about $\mbox{P}_{\mbox{\scriptsize V}}$ and $\mbox{P}_{\mbox{\scriptsize VI}}$? That is, can one find nonlinear wave patterns that can be predicted by poles of solutions to $\mbox{P}_{\mbox{\scriptsize V}}$ and $\mbox{P}_{\mbox{\scriptsize VI}}$ (note that rational solutions to $\mbox{P}_{\mbox{\scriptsize V}}$ and $\mbox{P}_{\mbox{\scriptsize VI}}$ have been known, see \cite{ClarksonNIST, Clarkson_P5})?  This is another interesting question that lies outside the scope of this paper and will be left for future studies.

\section*{Declaration of competing interest}
The authors declare that they have no known competing financial interests or personal relationships that could have appeared to influence the work reported in this paper.

\section*{Acknowledgment}
The work of B.Y. was supported in part by the National Natural Science Foundation of China (Grant No. 12431008, 12201326).

\end{document}